\documentclass[letterpaper,twocolumn,10pt]{article}
\usepackage{usenix2019_v3}

\usepackage{lineno,comment}

\usepackage{pifont}
\usepackage{epsfig,epsf,url}
\usepackage{tabularx}
\usepackage{float}
\usepackage[linesnumbered,boxed,ruled,vlined]{algorithm2e}
\usepackage{amsmath}
\usepackage{amssymb}
\usepackage{mathtools}
\usepackage{amsthm}
\usepackage{rotating}
\usepackage{diagbox}
\usepackage{makecell}
\usepackage{algpseudocode}
\usepackage{multirow}
\newtheorem{theorem}{Theorem}
\newtheorem{definition}{Definition}

\long\def\comment#1{}
\usepackage{fancyvrb}
\usepackage{booktabs}
\usepackage{color}
\usepackage{subcaption}
\usepackage{caption}
\usepackage{listings}
\newcommand{\eg}{{\it e.g.},~}

\newcommand{\etc}{{\it etc.~}}
\newcommand{\ie}{{\it i.e.},~}

\newcommand{\sys}{{\textsf{Wormhole}}\xspace}
\newcommand{\sysU}{{\textsf{Wormhole+Unison}}\xspace}

\newcommand{\paraspace}{\vspace{0.05in}} 
\newcommand{\topparab}[1]{\noindent{\bf #1}}

\newcommand{\parab}[1]{\paraspace\noindent{\bf #1}}

\newenvironment{icompact}{
  \begin{list}{$\bullet$}{
    \parsep 0.5pt plus 0.5pt
    \itemsep 0.5pt plus 0.5pt
    \leftmargin 0.2in}
       }
  {\normalsize\end{list}}

\newcount\hour \newcount\minute
\hour=\time  \divide \hour by 60
\minute=\time
\loop \ifnum \minute > 59 \advance \minute by -60 \repeat
\def\drafttime{\ifnum \hour<13 \number\hour:%
                      \ifnum \minute<10 0\fi
                      \number\minute
                      \ifnum \hour<12 \ AM\else \ PM\fi
         \else \advance \hour by -12 \number\hour:%
                      \ifnum \minute<10 0\fi
                      \number\minute \ PM\fi}

\usepackage{cleveref}
\crefname{section}{}{\S\S}
\crefdefaultlabelformat{\S#2#1#3}
\usepackage{hyperref}

\begin{document}
\title{Supercharging Packet-level Network Simulation of Large Model Training \\via Memoization and Fast-Forwarding}


\author
{
    \rm Fei Long$^{\star}$, Kaihui Gao$^{\ddag}$, Li Chen$^{\ddag}$, Dan Li$^{\star}$, Yiwei Zhang$^{\star}$, Fei Gui$^{\ddag}$, 
    \\ \rm Yitao Xing$^{\dagger}$, Wenjia Wei$^{\dagger}$, Bingyang Liu$^{\dagger}$\\
    \textit{$^\star$Tsinghua University\ \ \ ~~~~~\ \ $^\ddag$Zhongguancun Laboratory\ \ \ ~~~~~\ \ $^\dagger$Huawei Technologies Co., Ltd.}
    \\ \textbf{This manuscript is accepted by NSDI 2026.}
}




\maketitle


\begin{abstract}
Packet-level discrete-event simulation (PLDES) is a prevalent tool for evaluating detailed performance of large model training. Although PLDES offers high fidelity and generality, its slow performance has plagued networking practitioners.
Existing optimization techniques either simplify the network model, resulting in large errors; or execute it in parallel using multiple processors, with an upper bound on speedup.


This paper explores an alternative optimization direction that reduces the computational loads of PLDES while maintaining high fidelity. 
Our key insight is that, in distributed LLM training, packet-level traffic behaviors often exhibit \textit{repetitive contention patterns} and \textit{steady-states} where flow rates stabilize, ignoring these redundant discrete events speeds up the simulation considerably and the error is negligible.
We realize this idea by proposing \sys, a user-transparent PLDES kernel capable of automatically memoization for unsteady-states and skipping for steady-states. 
\sys adopts network partitioning, state memoization and reuse, and rate-based steady-state identification to accurately determine the periods of each flow's steady-state, while maintaining simulation consistency after fast-forwarding. 
Experiments demonstrate that \sys can achieve a 744$\times$ speedup over the original ns-3 (510$\times$ for MoE workload), with a bounded error of $<$1\%. Applying current multithreading parallel techniques and \sys together allows a 1012$\times$ speedup, reducing the simulation time for one GPT-13B training under 128 GPUs from 9 hours to 5 minutes.
\end{abstract}






\setlength{\textfloatsep}{1pt}
\setlength{\abovecaptionskip}{1pt plus 0pt minus 0pt}
\setlength{\belowcaptionskip}{1pt plus 0pt minus 0pt}
\medmuskip=0mu\relax 
\thinmuskip=0mu\relax 
\thickmuskip=0mu\relax 

\section{Introduction}
\label{sec:intro}
Large-scale infrastructures of Large Language Models (LLMs), serving as the engine behind the development of generative AI, are currently witnessing an unprecedented surge in investments~\cite{jiang2024megascale,gangidi2024rdma,Stargate}. 
LLM training simulation stands as a crucial means to ensure the most effective use of investment~\cite{won2023astrasim2,simai,multiverse,shen2025atlahs,qian2025miniature,kundu2024performance,feng2024echo,bang2024vtrain,lu2023distsim}, these simulators have become indispensable tools to explore the design space of model operator orchestration, parallel strategies, collective communication parameters, transport protocols, network topologies, \textit{etc.}

The packet-level discrete event simulation (PLDES)~\cite{ns2,ns3,omnet,dons,bai2024unison,mimicnet} is the predominant tool for LLM training simulations due to its high fidelity.
PLDES meticulously executes the behavior of each packet to accurately simulate performance critical events such as queueing, packet loss, and computation-communication overlapping; these events are important for the design of LLM training systems.

However, detailed simulation leads to heavy computational load, which has been a key problem for PLDES in LLM training simulations~\cite{multiverse,mimicnet}.
When simulating LLM training clusters, the network interconnects $10^3$--$10^6$ GPUs~\cite{jiang2024megascale,hu2024characterization,qian2024alibaba,gangidi2024rdma}, and carries a large number of elephant flows ($GB$ level), which generate a massive number of discrete events (>$O(10^{12})$) that should be executed strictly chronologically. As a result, existing PLDESs (\eg ASTRA-sim~\cite{rashidi2020astrasim,won2023astrasim2} with ns-3~\cite{ns3}) typically take several \textit{weeks} to simulate one training iteration of GPT3-175B~\cite{gpt3,won2023astrasim2}(\cref{sec:bg:expr}).



\parab{State of the arts.}
Currently, there exist two categories of technologies that aim to optimize the speed of the network PLDES. 
The first category involves coarse-grained modeling for networks. Specifically, flow-level simulators~\cite{alizadeh2011analysis,marsan2005using,misra2000fluid,peng2014multipath,ciucu2012perspectives,netcal,queueingtheory} and analytical models~\cite{lin2024towards, kundu2024performance} that solely calculate flow rates at network stability; 
AI-based methods employ deep neural networks to approximate network performance at different granularities: from individual switches~\cite{yang2022deepqueuenet}, to subtopologies~\cite{m3, mimicnet, kazer18fast}, and up to the entire network~\cite{rusek2020routenet,ferriol2023routenet}.
Both ignore key packet-level events that contribute to performance fluctuations, showing an error margin of 10\% to 25\% in the LLM training scenario. 

The second category enables parallel execution of PLDES using multiple cores or machines, which has recently received significant attention~\cite{dons,bai2024unison,multiverse,simai}. While it achieves a certain speedup, the acceleration effects exhibit sublinear scaling as the number of CPU cores increases, eventually reaching an upper bound. Experimental results (\cref{sec:bg:expr}) suggest that employing Unison~\cite{bai2024unison} to simulate GPT3-175B can achieve a maximum speedup of 10$\times$, after which the acceleration decreases as the number of cores increases.


\parab{Research question.}
To further speed up PLDES, this paper addresses the question: \emph{Can we reduce the computation load within PLDES while maintaining packet-level modeling and high fidelity?} 
A key advantage of reducing computational load of PLDES is in its orthogonality to existing multi-core parallelization techniques, enabling the acceleration of both approaches to compound.


\begin{figure}[t]
	\centering
	\includegraphics[width=0.98\columnwidth]{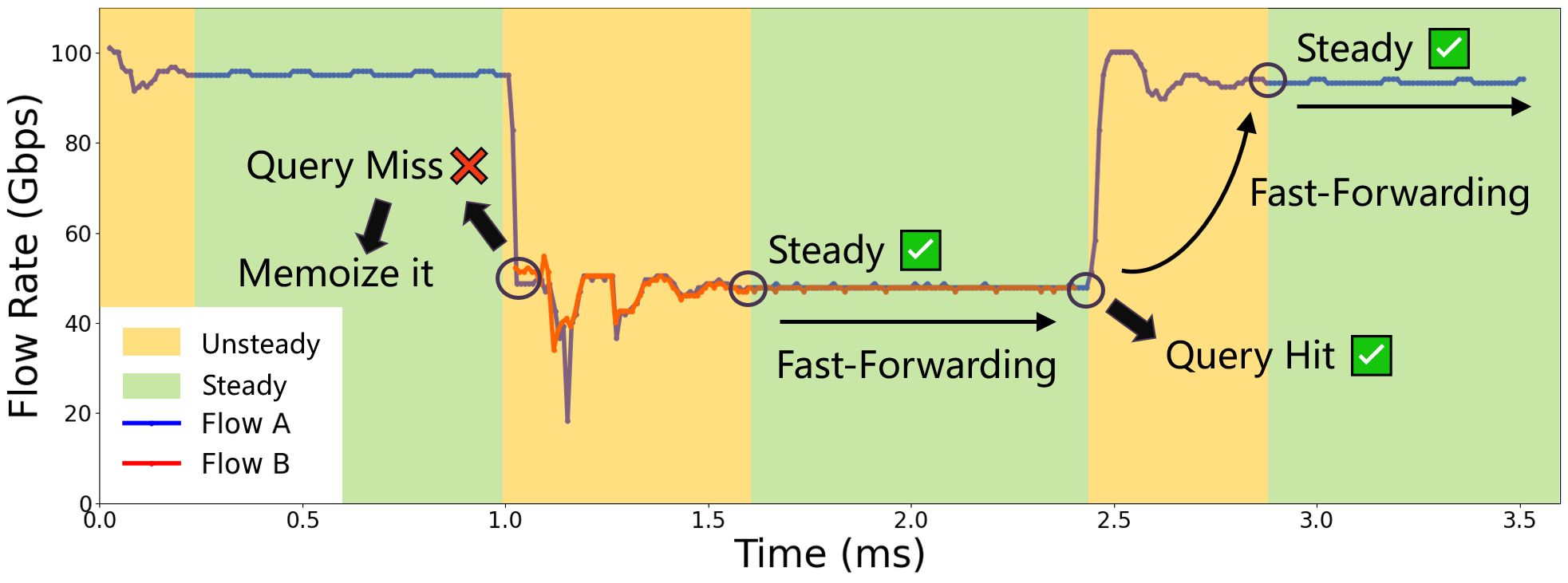}
	\caption{Illustrating unsteady-states, steady-states, state memoization, and simulation fast-forwarding in packet-level network simulation at the scenario of LLM training.}
	\label{fig:bg:flow_rate}
\end{figure}

\parab{Key insight.}
By analyzing the network communication patterns during LLM training, we observe that there are two categories of simulation load that can be reduced. 
(1) \emph{Repeated Contention Patterns}: identical traffic conflict patterns (\eg last-hop incast, flow contention at core switch) result in reproducible rate evolution dynamics. The simulation of these episodes can be fast-forwarded by reusing historical results. (2) \emph{Steady-state}: after congestion control algorithms (CCAs)~\cite{dctcp,dcqcn,timely,hpcc} convergence, flows enter a steady-state wherein transmission rates exhibit only minor periodic oscillations or remain constant. The PLDES of this repetitive regime is computationally redundant and can be similarly skipped.

LLM training, including GPT~\cite{gpt3} and MoE~\cite{deepseek-v3} models, is prone to these two phenomena, which is particularly revealed by the behavior of data parallel (DP) flows. 
Specifically, during one iteration of LLM training, DP flows (>$GB$ level) involve the synchronization of model parameters and gradients within different clusters, which may run periodically and converge to a stable rate.
In Figure~\ref{fig:bg:flow_rate}, we take two flows in one path during LLM training as an example to show fast-forward simulation.
In the unsteady-state with rate fluctuation, if the exact same contention pattern has occurred in the history, we can reuse the historical data to fast-forward the simulation process to the steady-state; if it has not occurred in the history, then PLDES is executed as usual.
Then, the two flows will go to steady-state, their rates tend to stabilize, and the queue length on the links also stabilizes without packet loss. 
Unless a new disruptive event (\eg flow enter/exit) occurs, this steady-state will continue.
The simulation can fast-forward to the next unsteady-state without significantly affecting the simulation results, the FCT error if we skip the steady periods in Figure~\ref{fig:bg:flow_rate} is less than 1\%.


\parab{Challenges.}
However, fast-forwarding both unsteady and steady-states to accelerate PLDES is not straightforward. 
Firstly, it requires the precise identification of the network regions in different states. If the whole network is regarded as a single region, it is difficult to match historical data and enter the global steady-state.
Secondly, memoizing unsteady-states necessitates the extraction of key features that characterize flow contention patterns. This feature set must be representative enough and reproducible enough to enable the frequent reuse of historical data without compromising the accuracy of the simulation. 
Finally, defining and identifying the steady-state involves a delicate trade-off between the accuracy and the degree of acceleration achieved --- a balance that is best guided by theoretical frameworks.





\parab{Our solution.}
This paper presents \sys, a user-transparent PLDES kernel capable of automatically memoization~\cite{michie1968memo, ford2002packrat, ford2004parsing} and fast-forwarding, compatible with current PLDES  parallelization technologies~\cite{bai2024unison,dons,fujimoto1990parallel}.
The main goals of \sys are to accurately identify and skip both the unsteady and steady-states, maintaining the consistency and correctness of the simulation. We propose three designs to address the aforementioned challenges.

\ding{182} Network Partitioning Algorithm: To accurately identify the network regions in different states, we observe that the network tends to form non-interfering partitions in LLM training, which can be handled separately.
To find them in advance, we propose a port-level network partitioning algorithm. This involves dividing the network into multiple connected graphs (\ie partitions) based on the ports through which flows go. 
Flows passing through the same port belong to the same partition. 
Consequently, the state of one partition is only determined by the flows within it.

\ding{183} Memoization for Unsteady-states: To utilize repeated flow patterns, we employ the memoization technique~\cite{michie1968memo, ford2002packrat, ford2004parsing} and build a database to record previously simulated scenarios. 
Specifically, we abstract a flow conflict graph (FCG) for one network partition, which captures the critical determinants of the unsteady-state processes. 
The database stores the first occurrence of unsteady-state processes, \ie the FCG (the \textit{key}) and the snapshot at the end of the unsteady-state (the \textit{value}).
In the subsequent simulation, when one network partition enters an unsteady-state, it will query the database and reuse the historical data (if the query hits).

\ding{184} Steady-state Identification Algorithm: By studying the dynamic equations of mainstream CCAs~\cite{hpcc,dcqcn,timely,dctcp}, we propose a steady-state identification algorithm based on a unified metric --- \textit{sending rate}. 
If the rate fluctuation within a monitoring interval is below a predefined threshold, the flow enters a steady-state, and the average rate within this interval is utilized as the rate during the steady-state.
Theoretical analysis confirms that when the sending rate is stable, other flow metrics are stable.
Moreover, we analyze the error of the algorithm as well as the choice of its parameters.


We prototype \sys based on ns-3~\cite{ns3} and address several practical challenges. 
1) To maintain the impact of steady regions on other flows, such as buffer occupancy, we pause the packets in the steady region's ports, thereby keeping the queue length constant until the steady-state ends. 
2) To elegantly skip the simulation process of ns-3 without reconstructing its underlying architecture, when skipping a period for a partition, we increase the timestamps of the partition's events by $\Delta T$, instead of clearing these events.

More specifically, we make the following contributions:

\begin{icompact}
     \item We find that there are two types of computational load that can be omitted in LLM training simulation, resulting in substantially faster simulation with negligible error, and we are the first to introduce the concepts of \textit{memoization} and \textit{fast-forwarding} into PLDES.
     
     \item We propose a series of algorithms to identify as many fast-forward opportunities as possible, notably network partitioning algorithm, steady-state identification algorithm, and simulation skip mechanism.
     
     \item Based on the dynamic equations of CCAs, we theoretically analyze the bounded simulation error of \sys and threshold guidance.

     \item We perform extensive experiments to confirm the speed and accuracy of \sys, which can achieve a 744$\times$ speedup on simulating GPT3-175B workload (510$\times$ on MoE workload) over the original ns-3, with a bounded error of $<$1\%. Applying Unison and \sys simultaneously allows a 1012$\times$ speedup on GPT3 workload (716$\times$ on MoE workload).
     
\end{icompact}
\parab{Limitations.} 
In highly dynamic and random-flow-pattern scenarios such as public cloud and multi-tenant workloads~\cite{wang2015bandwidth, yu2017survivable, li2021analyzing}, the traffic exhibits fewer repeating patterns and reaches a steady-state at lower frequency. 
Consequently, the benefit of \sys diminishes, with performance degrading to ns-3 baseline levels in worst-case scenarios, but without extra time cost and accuracy loss.

\noindent
\textit{This work does not raise any ethical issues.}

\begin{figure*}[t!]
	\centering
	\begin{minipage}{.32\linewidth}
		\subfloat[{Speed for ns-3 to simulate network communications of LLM training}]{
			\includegraphics[width=\columnwidth]{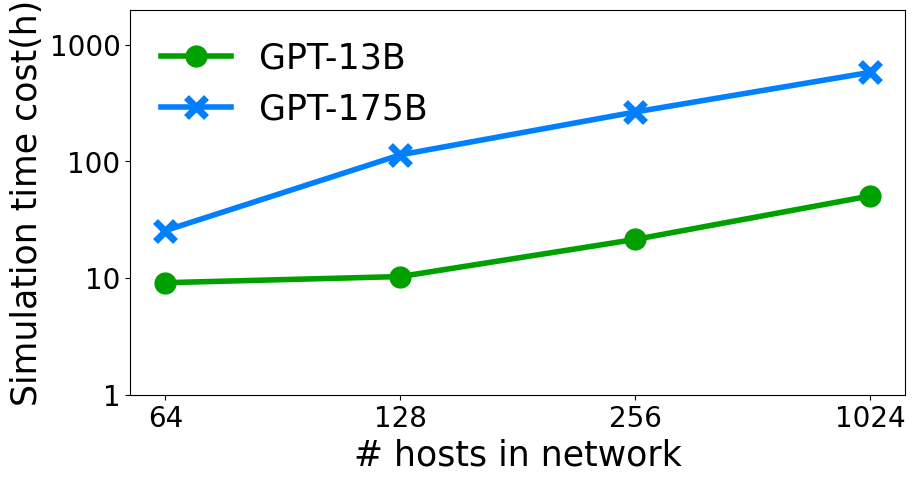}
			\label{fig:bg:des_timecost}
		}
	\end{minipage}
\hfill
	\begin{minipage}{.32\linewidth}
		\subfloat[{Upper bound on the speedup of the multi-threaded simulation approach}]{
			\includegraphics[width=\columnwidth]{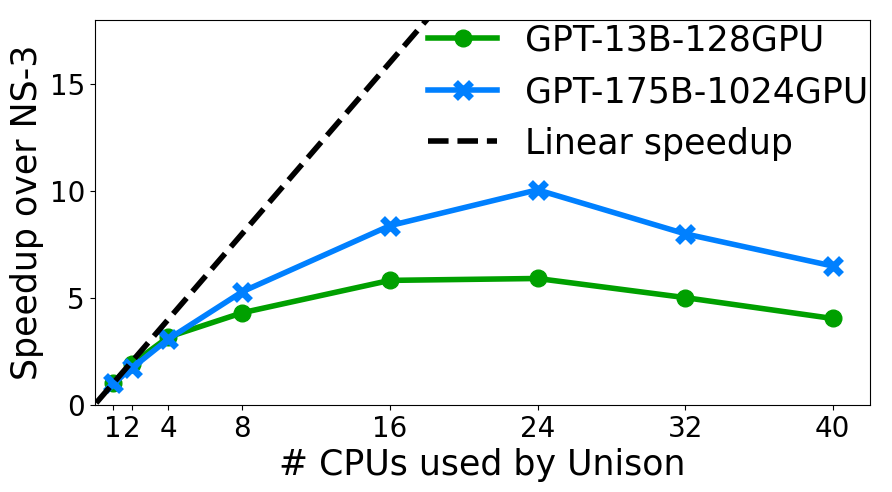}
			\label{fig:bg:unison_sublinear}
		}
	\end{minipage}
 \hfill
        \begin{minipage}{.32\linewidth}
		\subfloat[{Error of flow-level simulators and AI-based methods}]{
			\includegraphics[width=\columnwidth]{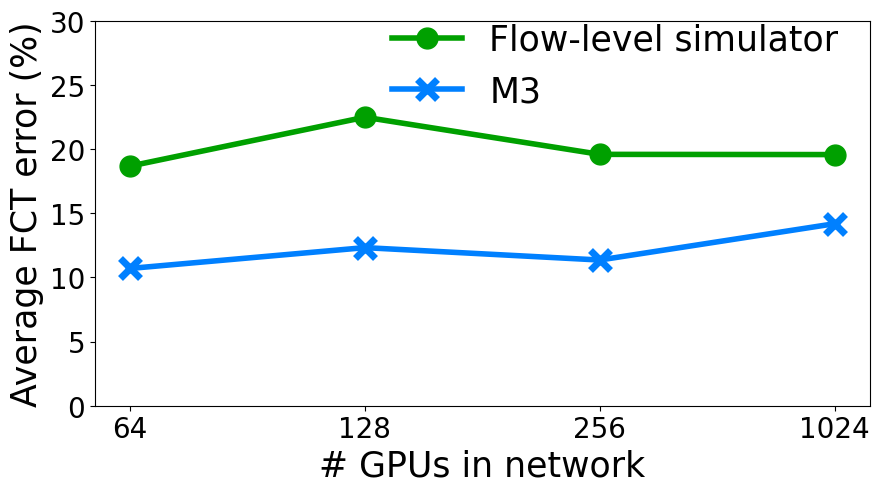}
			\label{fig:bg:flow_ai}
		}
	\end{minipage}
	\caption{The efficiency and accuracy of existing network simulation techniques for LLM training.}
	\vspace{-0.1in}
	\label{fig:eval:profiler}
\end{figure*}

\section{Background and Motivation}
\label{sec:bg}

In this section, we first present the background of LLM training simulation, and examine the performance of existing simulators. 
We then analyze the traffic characteristics of LLM training, and motivate the design of \sys. 

\subsection{Network Simulation in LLM Training}
\label{sec:bg:expr}



\parab{Backgrounds.}
In recent years, the scale of LLM parameter changes rapidly to hundreds of billions and even trillions~\cite{gpt3,switch-transformers}.
LLM training is conducted in a distributed manner across multiple high-performance computing nodes to expedite model convergence. 
To achieve efficient parallel training of LLMs, data parallelism (DP)~\cite{hillis1986data,pytorch-distributed}, tensor parallelism (TP)~\cite{megatron-lm-1}, pipeline parallelism (PP)~\cite{gpipe,megatron-lm-2}, sequence parallelism (SP)~\cite{li2023sequence, megatron-lm-3} and expert parallelism (EP)~\cite{switch-transformers, deepseek-v3} are commonly employed as parallel acceleration methods. 
These methods generate multiple communication domains that perform collective operations according to the workload requirements of LLM training.

Thus, the efficiency of network communication significantly affects training performance. 
Typically, network simulation technology is required to precisely and cost-effectively optimize network communication, identifying the most suitable network configurations for LLM training.
However, as the parameter count of LLMs increases, the scale of training clusters has also expanded to $10^3$ --$10^6$ GPUs~\cite{jiang2024megascale,hu2024characterization,qian2024alibaba,gangidi2024rdma}. 
The performance of current state-of-the-art network simulators is difficult to meet requirements. Next, we conduct experiments to confirm this.

\parab{Experiment settings.}
To compare the speed and accuracy deficiencies of state-of-the-art methods, we simulate a LLM training workload iteration with DP and PP traffic on a 56-core server. The LLM training network is set to scale from hundreds to thousands of GPUs, generating $O(10^2)-O(10^4)$ of DP and PP flows.

\parab{Packet-level discrete-event simulation.}
The pure single-process packet-level simulation in ns-3 is full-fidelity but extremely inefficient. 
During the simulation process, a discrete event level of $O(10^{12})$ can be generated. 
As shown in Figure~\ref{fig:bg:des_timecost}, with the increase in cluster scale, the time cost also shows an exponential growth trend, and it takes several weeks to complete the simulation of large-scale clusters.

\parab{Parallel and distributed DES.}
UNISON~\cite{bai2024unison} and DONS~\cite{dons} use multithreading to execute PLDES in parallel on multiple CPUs.
But they can only provide sublinear speedups due to the synchronization overhead, which always hits an upper bound, as shown in Figure~\ref{fig:bg:unison_sublinear}. 
Another downside is that it uses up a lot of CPU cores that could be used to run multiple independent experiments in parallel, which would exhibit a linear speedup, since LLM engineers always run multiple sets of experiments at once.

\parab{Flow-level simulation.}
Flow-level methods typically fall into two categories. The first computes stable flow rates via max-min or waterfilling allocation~\cite{alizadeh2011analysis,marsan2005using,misra2000fluid,peng2014multipath,ciucu2012perspectives,netcal,queueingtheory}. The second uses analytical models~\cite{lin2024towards, kundu2024performance} or profiling-based techniques~\cite{feng2024echo, bang2024vtrain, lee2024data, lu2023distsim} to approximate flow completion times. Both ignore packet-level dynamics such as queueing, congestion control, and transient losses, and producing $\sim$20\% FCT error under dynamic LLM workloads (Figure~\ref{fig:bg:flow_ai}).

\parab{AI-based methods.} 
AI-based methods~\cite{yang2022deepqueuenet,kazer18fast,mimicnet,rusek2020routenet,ferriol2023routenet} are designed to achieve faster performance estimation using neural network models. 
However, the acquisition of training data incurs high costs. 
Without a sufficiently large training set, they may not generalize to network scenarios.
In Figure~\ref{fig:bg:flow_ai}, M3~\cite{m3} has an error of 10-15\% in various scenarios, which is lower than that of flow-level simulators but still results in significant absolute deviations.

\parab{Summary.} Existing PLDES optimization techniques either simplify the network model, resulting in large errors, or execute it in parallel using multiple processors, with an upper bound on speedup. 
Practitioners are in urgent need of an optimization technique that reduces the computational burden of PLDES and maintains accuracy.

\subsection{Repeated Flow Patterns in LLM Training}\label{subsec:repeat_flow}
\label{sec:bg:repeated}
\begin{figure}[t]
	\centering
	\begin{minipage}{.45\linewidth}
		\subfloat[Repeated contention patterns]{
		\includegraphics[width=\columnwidth]{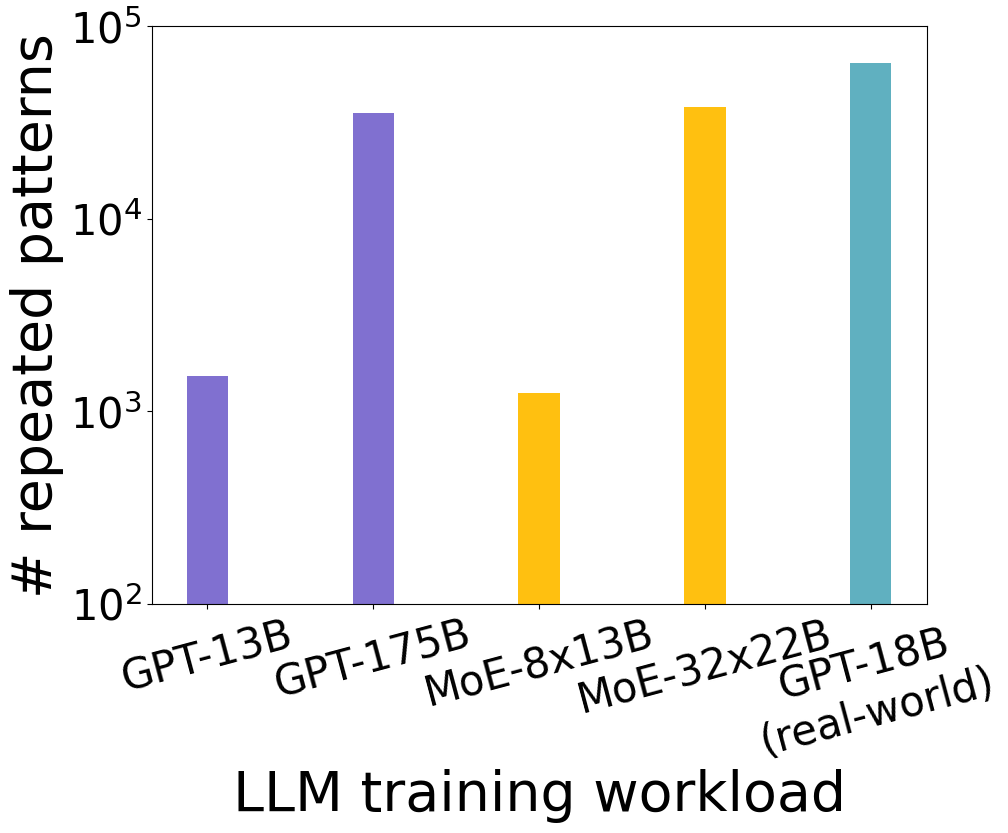} 
		\label{fig:bg:pattern}
		}
	\end{minipage}
	\begin{minipage}{.5\linewidth}
		\subfloat[Proportion of steady-states]{
		\includegraphics[width=\columnwidth]{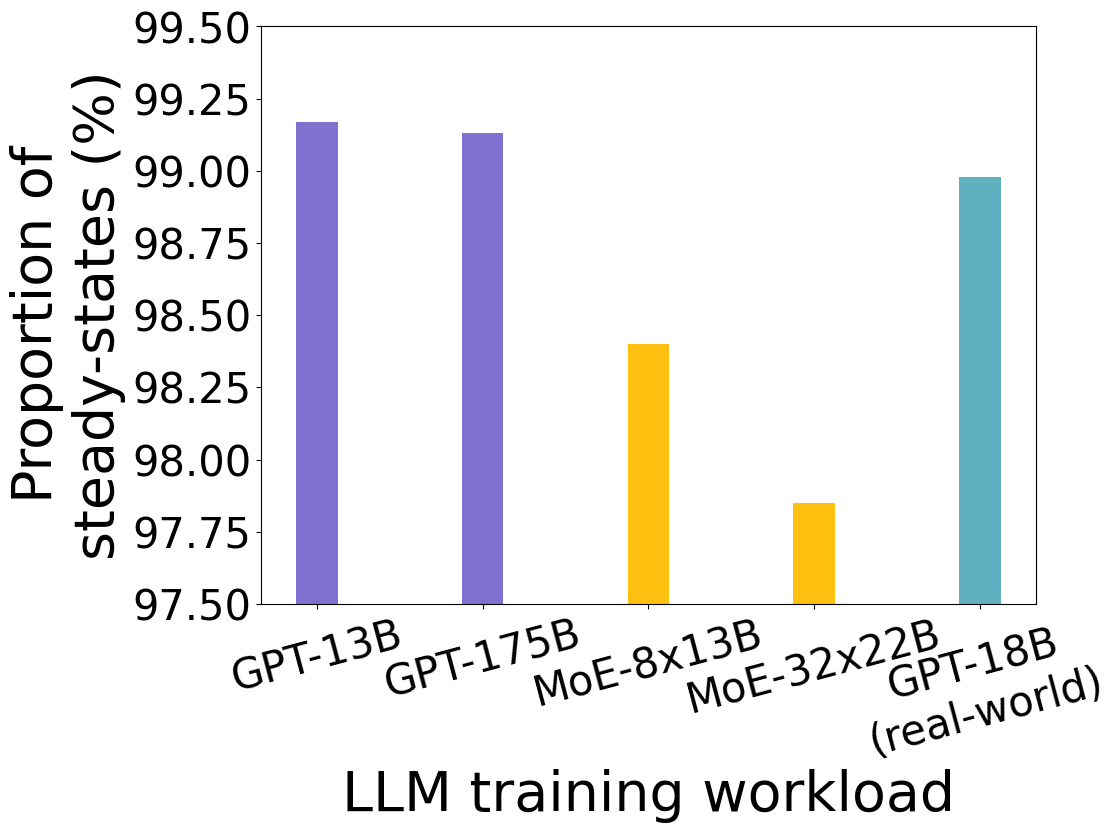} 
		\label{fig:bg:steady}
		}
	\end{minipage}
	\caption{ Repeated contention patterns and steady-states in LLM training.}
	\label{fig:bg:memoization}
\end{figure}

In large-scale model training, flow contention is a common occurrence due to the presence of diverse parallelization strategies.
For instance, All-reduce flows in DP and point-to-point flows in PP may encounter contentions.
A flow contention pattern is described by the set of flows that have overlap, as well as the links they travel through.
We observe that the contention patterns are highly repetitive in LLM training. 
Under a fixed parallelization strategy, the same collective communication tasks are invoked multi-times in one training iteration, and thus the same contention patterns reappear in LLM training simulation. 

To verify the prevalence of the repeated flow pattern, we conduct simulations on SimAI~\cite{simai} for the training of GPT-13B and MoE-8$\times$7B on 128 GPUs, as well as experiments for GPT-175B and MoE-32$\times$22B on 1024 GPUs, 

We also analyze a real world trace of the training of a GPT-18B on 256 NVIDIA A100 GPUs. 
All these experiments employ the Rail-Optimized Fat-tree~\cite{nvidia2023superpod} topology.

As Figure~\ref{fig:bg:pattern} shows, training a GPT-13B or MoE-8$\times$7B on 128 GPUs yields flow pattern repetitions over 1200 times per iteration, where a total of 1633 flow contention patterns are identified. 
Training a GPT-175B or MoE-32$\times$22B on 1024 GPUs yields nearly 40,000 repetitions. 
The real-world case of training a GPT-18B on 256 GPUs yields 65,870 flow contention pattern instances, which collapse into 1,488 distinct patterns with over 60,000 redundant occurrences.

\subsection{Steady-state in LLM Training}






\parab{Proportion of steady-states in LLM training.}
A steady-state in LLM training workloads refers to a temporal interval (steady period) during which network flows in a specific topological subset (steady region) exhibit stable transmission behaviors. 
LLM training uses efficient parallelization methods that generate a significant number of elephant flows of the same size within the cluster~\cite{geng2019elasticpipe, wang2018bml}.
Specifically, when employing data parallelism in LLM training, DP communication domains produce DP flows reaching GB levels. 

Based on the above simulation experiments, we verify and quantify the existence of steady-states in LLM training. We defer the formal definition of the steady-state to $\S$\ref{sec:definition:steady_state}.
As shown in Figure~\ref{fig:bg:steady}, dense models (\eg GPT-175B) exceed 99\% steady-state proportion, while MoE models, due to all-to-all traffic in EP, exhibit $\sim$97.5\%; real-world traces demonstrate 98.82\% steady-state proportion.
We attribute the prevalence of steady-states in LLM training to the periodic exchange of parameters or gradients at fixed intervals required by parallelization methods and the constraint of flow paths to repetitive topological subsets by collective communication patterns. 


\parab{Numerical analysis of simulation error.}
For the case in Figure~\ref{fig:bg:steady}, we obtain the rate evolutions of all flows and perform a numerical analysis to compute the error of FCT after skipping steady-states. 
We identify the steady-state offline and use the average rate over this period as the constant sending rate during the steady-state.
The speedup is calculated by dividing the total flow size by the amount of data sent during the steady period. 
The relative error of the FCT is determined by comparing the FCT derived from the original data with the estimated FCT.
The results indicate that, by ignoring all events in steady-states, the simulation of LLM training can be accelerated by 120$\times$ on GPT (60$\times$ on MoE), with an average FCT error of only 1\%. This high acceleration potential motivates \sys.







\section{\sys Overview}\label{sec:overview}
In this section, we first define two important terms and concepts, \textit{network partition} and \textit{steady-state}, and then introduce the overall workflow of \sys.
We present the skipping of unsteady-states in $\S$\ref{sec:memoization} and skipping of steady-states in $\S$\ref{sec:steady}.

\subsection{Terminology \& Formal Definitions}

\subsubsection{Port-level Network Partition}
\label{sec:definition:partition}

Data center networks are large in scale and have multiple paths, which are shared by a large number of traffic flows~\cite{li2014exr, li2014willow}. 
If the network is treated as a whole when identifying the steady-state, the network will enter the steady-state only after all flows are stable, which is rare. 
Nevertheless, we observe that due to the locality nature of task deployment in data centers, the paths traversed by flows may not interfere with each other, forming multiple sub-networks.

In the scenario of LLM training, the deployment of DP and TP leads to the formation of several non-intersecting data parallel communication groups and tensor parallel communication groups within the network.
Previous work~\cite{simai} has also employed packet-level DES to simulate TP flows. 
The links between DP groups and TP groups do not overlap, as DP groups are generally deployed across different Points of Delivery (PoDs), while TP groups are deployed within a high-bandwidth domain. 
Moreover, both types of communication are confined within their respective communication domains. 
Consequently, the communication within DP groups and TP groups can be considered as occurring in different relatively independent sub-networks.
Obviously, the states of these sub-networks do not affect each other, and identifying their states separately will avoid the steady-state being missed.

To capitalize on this characteristic, we formulate a formal definition of \textit{network partitioning}.

\begin{definition}[Network Partition]
Flows sharing the same port (or link), along with the ports and links through which these flows traverse, construct a network partition.
\end{definition}

\begin{figure}[t]
	\centering
	\begin{minipage}{.65\linewidth}
		\subfloat[An example of network partitions]{
		\includegraphics[width=\columnwidth]{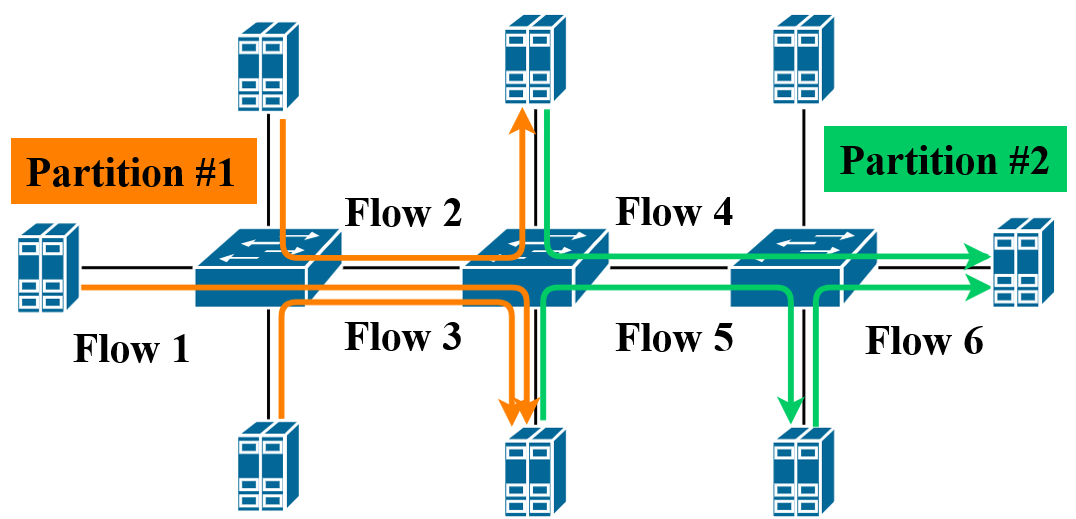} 
		\label{fig:partition}
		}
	\end{minipage}
	\hspace{0.1in}
	\begin{minipage}{.25\linewidth}
		\subfloat[Examples of Flow Conflict Graph]{
		\includegraphics[width=\columnwidth]{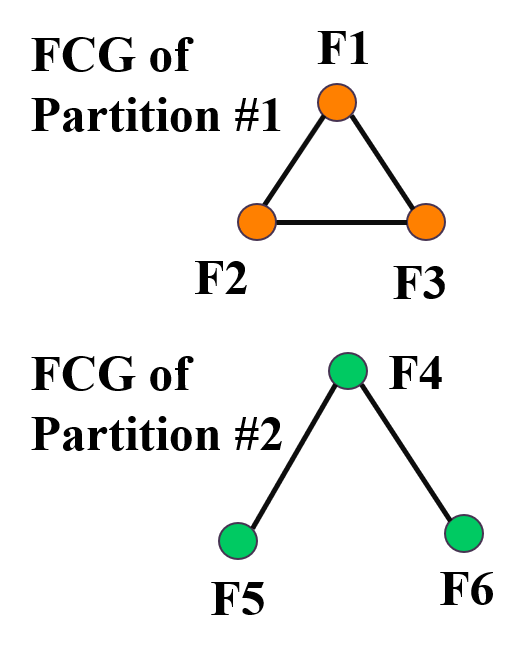} 
		\label{fig:FCG}
		}
	\end{minipage}
	\caption{Examples of network partitions and the corresponding Flow Conflict Graphs (FCGs).}
	\label{fig:partition_FCG}

\end{figure}

As demonstrated in Figure~\ref{fig:partition}, we segment flows and links so that no two flows in different partitions share common links. 
We define network partitions based on flow intersections at switch ports rather than entire switches. 
We reject switch-centric partitioning because individual ports on a switch exhibit minimal mutual interference despite co-location. 
By isolating flows that traverse shared ports into distinct partitions, our method achieves three key benefits: (1) finer-grained parallelism through contention-aware division, (2) elimination of hidden dependencies between logically separated flows, and (3) increased probability of finding steady-states.
Notably, in MoE training workloads, spatial locality of all-to-all traffic and typical EP degrees ($\leq$ 64)~\cite{switch-transformers, deepseek-v3} inherently constrain partition sizes no more than 128 GPUs.


\subsubsection{Steady-state}
\label{sec:definition:steady_state}

Training LLMs in data center networks generates numerous elephant flows with sizes larger than 1GB, leading to a substantial number of discrete events. 
Conducting packet-level simulation for all these discrete events results in an enormous computational overhead, rendering the simulation process inefficient. 
We discover that after being regulated by congestion control algorithms, the flow rates of these elephant flows can stabilize within a relatively narrow range, exhibiting a trend of gradual convergence or periodicity. 
Furthermore, we posit that skipping these events during the simulation process has a negligible impact on the simulation outcomes. 

\begin{figure}[t]
	\centering
	\includegraphics[width=0.8\columnwidth]{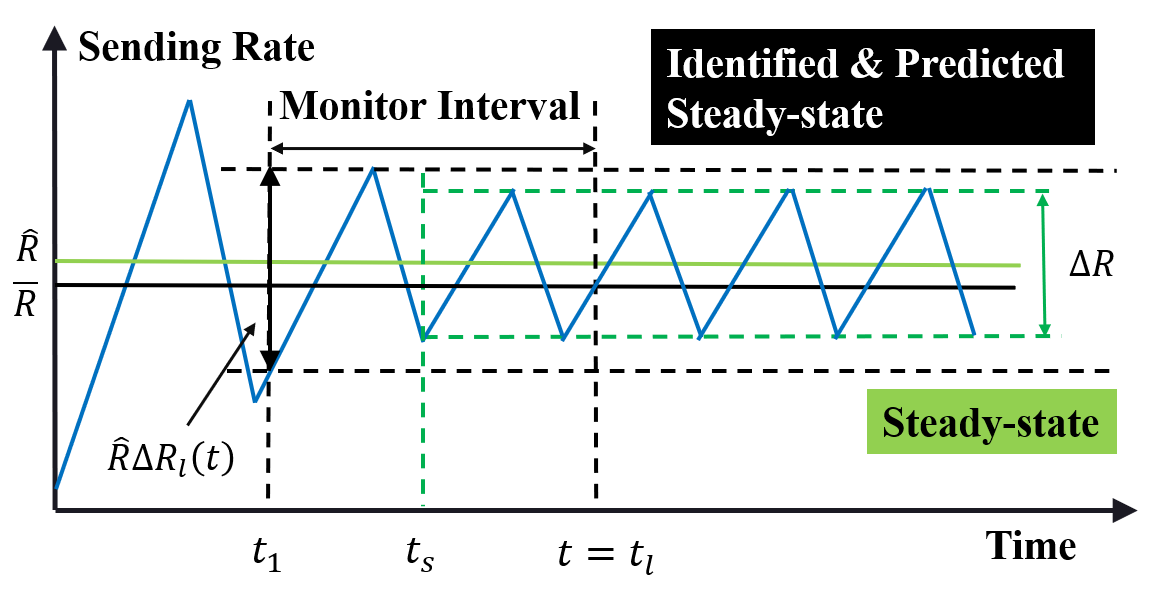}
	\caption{Flow sending rate during CCA convergence.}
	\label{fig:steady_state}
\end{figure}

To ensure the reliability and efficiency of the simulation process, the steady-state period of the flows must be both sustainable and predictable, characterized by repeatability and periodicity. 
When a flow is in its steady-state, various metrics associated with it should exhibit minimal fluctuations, confined within a narrow range.
Therefore, we formally define the \textit{steady-state} as follows:

\begin{definition}[Traffic Steady-state]
A flow is considered in a steady-state between times $t_s$ and $t_f$ if and only if the fluctuations of a set of key metrics are confined within a narrow range. 
Mathematically, this condition can be expressed as:
\begin{align}
\label{formula:metric_fluc}
\Delta X &= \max_{\substack{t_s \leq t \leq t_f}}\{X(t)\} - \min_{\substack{t_s \leq t \leq t_f}}\{X(t)\} \\
\label{formula:metric_stable}
\Delta X &< \epsilon_X
\end{align}
where $X$ represents a suite of flow-related metrics, including sending rate $R$, congestion window size $cwnd$, round-trip time RTT, queue length $Q$, and in-flight bytes $I$. 
Here, $\epsilon_X$ is a small constant threshold of the metric $X$ introduced to accommodate the inherent periodicity and minor oscillations observed in the metrics. 
\end{definition}
This definition of the steady-state captures the essential characteristics of flow behavior during periods of stability, ensuring that the metrics exhibit minimal variation over the specified interval. 
As shown in Figure~\ref{fig:steady_state}, the steady-state is considered to start at $t_s$ when $\Delta R$ reaches the threshold value.
Here, the time interval for defining the steady-state should be based on the convergence of the congestion control algorithms.
Additionally, a network partition is considered to be in steady-state if and only if all flows within it have attained steady-state; otherwise it is considered to be in \textit{unsteady-state}.

\subsection{Workflow}

\begin{figure}[t]
	\centering
	\includegraphics[width=0.98\columnwidth]{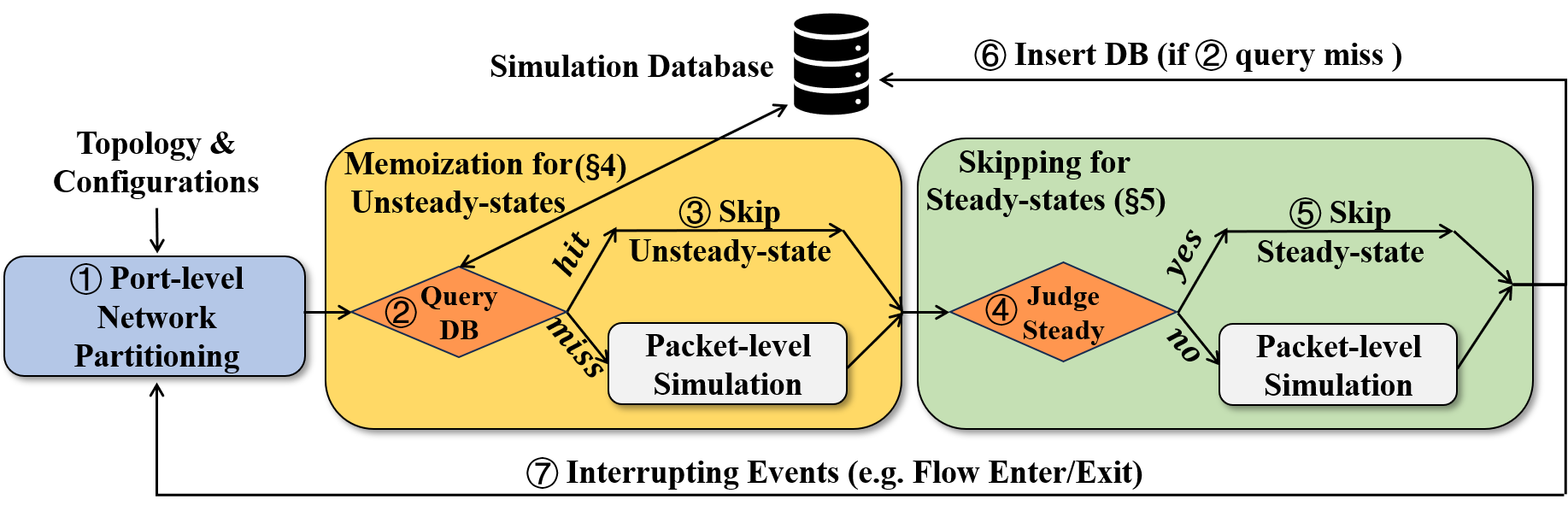}
	\caption{\sys's high-level workflow.}
	\label{fig:workflow_memoization}
\end{figure}

The workflow of \sys is illustrated in Figure~\ref{fig:workflow_memoization}.
Upon user input of the topology, configurations, \ding{172} \sys first performs the network partitioning at port level ($\S$\ref{sec:design:algorithm}), which divides the flow into several disjoint sets, facilitating subsequent processing.

Then, \ding{173} \sys queries a simulation database based on the current partition information to check for the existence of simulation results for the identical scenario. 
The simulation database adopts the principle of memoization~\cite{michie1968memo, ford2002packrat, ford2004parsing}, recording once-simulated scenarios and their key outcomes as reusable knowledge ($\S$\ref{sec:design:FCG_retrieve}). 
If such results are available, \ding{174} \sys bypasses the transient phase and transitions directly to the steady-state processing logic. 
Conversely, if no matching results are found, \sys proceeds with the packet-level simulation.

During the process of skipping the unsteady-state or employing packet simulation, \ding{175} \sys employs a steady-state identification algorithm to determine whether the system has reached a steady-state The simulation database adopts the principle of memoization, recording once-simulated scenarios and their key outcomes as reusable knowledge ($\S$\ref{sec:design:ident:ident}). 
If a partition enters a steady-state, \ding{176} \sys skips this steady-state period; otherwise, it continues with the packet-level simulation. 
When the steady-state of the network partition ends, or when the simulation of the network partition terminates at the packet-level granularity, \ding{177} if \sys misses in the previous query, it encounters a novel simulation scenario and is required to insert the key information of the current simulation situation into the database for subsequent simulation processes ($\S$\ref{sec:design:FCG_store}). 
Subsequently, \ding{178} \sys continues to process relevant interruptive events, including flow enter and flow exit, which alter the current network partitioning pattern($\S$\ref{sec:design:flow}). 
Therefore, the network partitioning algorithm needs to be invoked again to restart this cyclical process.



\section{Memoization for Unsteady-states}\label{sec:memoization}
In this section, we first present the definition of the \textit{Flow Conflict Graph (FCG)} as the central abstraction for representing network partitions in unsteady-states.
We then describe the mechanisms for storing unsteady processes using FCGs, followed by the procedure for and retrieving and bypassing unsteady-states through memoization.

\subsection{Network Partitioning Algorithm}
\label{sec:design:algorithm}

During the processing of interrupting events, the network partitions the entire network into several network partitions based on the topological overlap relationships of the flows.
We achieve this partitioning by constructing a bipartite graph where vertices represent flows and links, creating edges between a flow vertex and a link vertex if the flow traverses that link. The depth-first search (DFS)~\cite{DFS} is then used to identify connected components (partitions) within this graph. 

Appendix~\ref{appendix:network_partitioning_algorithm} describes the network partitioning algorithm used in \sys.
Given that each node is visited at most once, the complexity of the algorithm is $O(N+M)$, where $N$ represents the number of flows and $M$ represents the number of links in the network. 
When a new flow starts, some existing network partitions may be merged. 
When an existing flow finishes, an existing network partition may be split. 
In response to the dynamic reconstruction characteristics of network partitions, we have proposed an incremental network partitioning algorithm, which is detailed in Appendix~\ref{appendix:incremental}.

\subsection{Flow Conflict Graph}
\label{sec:design:FCG}

When a network partition operates in an unsteady-state, its flows are regulated by CCAs until the rates converge. 
Due to the diversity of partition topologies and the inherent complexity of CCA dynamics, the transient evolution of flow rates is analytically intractable. 
However, as discussed in \cref{sec:definition:partition}, the training traffic of large-scale learning models often exhibits strong locality and symmetry, resulting in recurring partition structures across simulations.
This recurrence motivates the use of memoization: the simulator can identify previously encountered cases and reuse the results instead of re-executing identical transient processes.

To support efficient detection and reuse of repeated cases, we need a lightweight yet expressive representation of simulation state. 
We define the \textit{Flow Conflict Graph (FCG)} as such an abstraction. 
FCG models the conflict relationships among flows within a partition as an undirected graph.
Vertices represent flows and edges denote link-sharing dependencies.
If two flows traverse at least one common link, an edge is placed between them. 
To encode transient information, the vertex weights correspond to instantaneous flow rates, while the edge weights represent the number of overlapping links between the corresponding flows. 
Note that FCG deliberately ignores the absolute path length of flows and their spatial positions in the topology as the resulting error is negligible.

This design ensures that FCGs capture both the structural and quantitative features of unsteady partitions, providing a formalized foundation for storage, retrieval, and reuse, which greatly reduces the storage footprint and accelerates retrieval. 
It creates a canonical representation that makes different simulation runs comparable, enabling reliable detection of recurring patterns.

\subsection{Storing Unsteady-states}
\label{sec:design:FCG_store}

Upon the initialization of a partition, the simulator constructs an FCG that characterizes its starting condition, which is denoted as $FCG_{start}$.
The steady-state identification algorithm (\cref{sec:design:ident:ident}) subsequently monitors the evolution of the partition and records a corresponding $FCG_{end}$ when a steady-state terminates or the partition encounters a flow completion event.

Simulation database then stores this process in a key–value format, where the key is the starting condition $FCG_{start}$ and the value is a structured tuple containing the essential results of the transient phase:
\begin{align}
key &: \ FCG_{start} \nonumber \\
value &: \ (FCG_{end}, \; \{Size_{f}\}|_{f\in F}, \; T_{conv}) \nonumber
\end{align}
Here, $FCG_{end}$ denotes the FCG at the end of the transient phase, $\{Size_{f}\}|_{f\in F}$ represents the aggregate transmission volume of flows $f$ in the partition during the unsteady period, and $T_{conv}$ is the measured convergence time.

Simulation database does not preserve the full temporal evolution of a network partition. 
Instead, it decomposes the process into unsteady and steady phases. 
This design follows the observation that flow sizes determine the duration of steady-states, whereas flow sizes remain independent of the transient convergence dynamics. 
Consequently, the database records only network snapshots at the entry and exit of unsteady phases, the steady-state transmission rate of each flow, and the associated convergence time. These quantities are sufficient for reconstructing FCTs at per-flow granularity within network simulations.



\subsection{Retrieving and Skipping Unsteady-states}
\label{sec:design:FCG_retrieve}

When a new partition is formed, the simulator constructs its $FCG_{start}$ and queries the simulation database to determine whether an equivalent unsteady process has already been recorded. 
The lookup process first eliminates candidates with mismatched structural characteristics (\eg different numbers of vertices or edges). 
For the remaining candidates, a weighted graph isomorphism procedure is applied~\cite{cordella2004sub, hagberg2008exploring}. 
If the simulation database returns a matching entry, the simulator bypasses the costly convergence phase of the CCAs.
This mechanism reserves accuracy by using graph-based matching that encodes both structural and quantitative information, ensuring that reused results are semantically equivalent to recomputed ones. 
As a result, this mechanism maintains the correctness while maximizing reuse opportunities.




\section{Fast-forwarding for Steady-states}\label{sec:steady}
In this section, we first present the steady-state identification algorithm.
We then provide an analysis of parameter errors and guidelines for hyper-parameters selection. 
Finally, we discuss the termination conditions of the steady-state.

\subsection{Steady-state Identification}
\label{sec:design:ident}

Training LLMs in data center networks generates numerous elephant flows, where simulating every packet-level event incurs enormous computational overhead. Notably, once regulated by congestion control, these flows typically converge to stable or periodic rates.
Therefore, identifying the steady-states of these flows and bypassing them during simulation can significantly reduce the computational time and resource consumption associated with the simulation process.

Next, we describe the metrics employed for steady-state identification and present a steady-state identification algorithm, along with the corresponding error analysis.

\subsubsection{Metrics for Steady-state Identification}
\label{sec:design:ident:metrics}

In practice, simultaneously calculating the fluctuation ranges of all metrics is redundant and unnecessary~\cite{cceval}.
Therefore, we select the sending rate $R$ as the indicator to assess the steady-state condition.
The theory indicates that these metrics are stable when the sending rate $R$ is stable, which is empirically validated through experiments in $\S$\ref{sec:eval:sensitive}.
Subsequently, we proceed to prove the following theorem.

\begin{theorem}
\label{theorem:rate-metrics}
Within the time interval [$t_s$, $t_f$], the congestion control algorithm converges. Assuming the flow rate is stable, \ie it satisfies $\Delta R < \epsilon_R$, then for $cwnd$, RTT, $Q$, and $I$, they are also stable.
\begin{align}
\label{formula:r_tot_stable}
\Delta R &< \epsilon_R, \quad t_s \leq t \leq t_f
\end{align}
This means when Equation ~\ref{formula:r_tot_stable} holds, there exist small constants $\epsilon_{cwnd}$, $\epsilon_{RTT}$, $\epsilon_Q$, and $\epsilon_I$ such that the following inequalities hold, according to the definition of the steady-state of the flow and the stability of the metrics from Equation ~\ref{formula:metric_fluc}:
\begin{align}
    \Delta cwnd < \epsilon_{cwnd}, \quad \Delta RTT < \epsilon_{RTT}, \quad \Delta Q < \epsilon_Q, \quad \Delta I < \epsilon_I
\end{align}
\end{theorem}

We provide the proof in Appendix ~\ref{appendix:proof_metrics}.
Theorem ~\ref{theorem:rate-metrics} indicates that when $R$ is stable, other metrics are stable as well, further corroborating the stability of CCAs~\cite{hpcc,dcqcn,timely,dctcp}. 

\subsubsection{Steady-state Identification Algorithm}
\label{sec:design:ident:ident}

In accordance with the definition of the steady-state, our goal is to construct an algorithm that can identify the steady-state during CCA converges, while minimizing the resulting error at the same time.
In Theorem ~\ref{theorem:rate-metrics}, we identify and derive the desirable property that multiple metrics simultaneously achieve stability in the steady-state, which enables us to construct a unified steady-state identification algorithm based on a single metric: the sending rate $R$.


The steady-state identification algorithm is designed to accurately determine whether the flow has entered a steady-state based on a short segment of local data. 
This approach necessitates that the flow is considered to be in a steady-state once the fluctuations in the local data stabilize, and it remains stable until the flow completion or the occurrence of other interrupting events.
Mathematically, this condition can be expressed as:
\begin{align}
\label{formula:stable_holds}
\Delta R &< \theta, \quad when \ \max\limits_{1 \leq k \leq l}\{R(t_k)\} - \min\limits_{1 \leq k \leq l}\{R(t_k)\} < \theta
\end{align}
where $\theta$ and $l$ are two appropriately set hyperparameters in advance, and $t_1 < t_2 < \dots < t_l = t$ are a series of time instances at which the sampling of transmission rates occurs.
Equation ~\ref{formula:stable_holds} holds as long as the CCA is convergent, as the CCA ensures that the flow rate remains constant over time or exhibits a sawtooth pattern upon convergence.
With the theoretical guidance provided by Equation ~\ref{formula:stable_holds}, steady-states can be accurately identified when $l$ and $\theta$ are reasonable.


\parab{Identification algorithm.}
To determine whether a flow is in a steady-state, we compare the fluctuation of the flow rate with a predefined threshold. 
By maintaining a fixed length rate detection interval $l$, we can collect flow rate data over a specified period, which allows us to calculate the fluctuation of the flow rate.
To enhance the criteria for determining the steady-state, we introduce a condition that the relative fluctuation must be smaller than a predefined threshold $\theta$.
For the rate detection interval $t_1 < t_2 < ... < t_l = t$, the fluctuation degree $\Delta R_l(t)$ of the flow rate is calculated as follows:
\begin{align}
\label{formula:relative_fluc}
\Delta R_l(t) &= \frac{\max\limits_{1 \leq k \leq l}\{R(t_k)\} - \min\limits_{1 \leq k \leq l}\{R(t_k)\}}{(\sum\limits_{k=1}^{l}R(t_k)) /l}
\end{align}
When the condition $\Delta R_l(t) < \theta$ is satisfied for the predefined threshold $\theta$, the flow can be identified as having entered the steady period, as shown in Figure ~\ref{fig:steady_state}.

\parab{Rate estimation.}
After the steady-state has been identified, we should estimate the rate of the flow in the steady-state, which will affect the accuracy of the FCT. 
The max-min fair rate allocation algorithm~\cite{jaffe1981bottleneck} cannot be used because in multi-hop congestion scenarios, the converged flow rates may deviate from max-min fairness as revealed by a previous work~\cite{wang2023poseidon}.
The criterion of PLDES is to correctly simulate the unique process of various algorithms, rather than simulating the most ideal network situation. 
We set the rate during the steady period as the average rate over this interval:
\begin{align}
\label{formula:estimated_r}
\hat{R} &= \frac{\sum\limits_{k=1}^{l}R(t_k)}{l}
\end{align}
The finish time of the steady period is described in \cref{sec:design:flow}.





\subsection{Error Analysis and Threshold Guidance}
\label{sec:design:ident:error_analysis}

The steady-state identification algorithm utilizes $\Delta R_l(t)$ to approximate $\Delta R$, and $\hat{R}$ to approximate the real steady-state average rate $\overline{R}$, which may introduce errors in the rate.
We now analyze the upper bounds of the errors resulting from these approximations.

\parab{Error of estimating the sending rate.}
We first consider the error in estimating the real steady-state average rate $\overline{R}$ using Equation ~\ref{formula:estimated_r} within the steady-state interval by the following theorem.
\begin{theorem}
\label{theorem:err_r_bound}
The relative error of stable flow rate estimation is bounded. Specifically,
\begin{align}
\hat{\epsilon_{R}} &= |\frac{\hat{R}-\overline{R}}{\overline{R}}| < \frac{\theta}{1 - \theta} \quad when \ \Delta R_l(t) < \theta
\end{align}
\end{theorem}

We provide the proof in Appendix ~\ref{appendix:proof_error_rate}.

\parab{Error of estimating the steady period duration.}
We proceed to analyze the error in estimating the flow completion time. 
In the absence of other interrupt events, the flow will transmit packets at a stable rate after entering the steady-state until all packets of the flow have been sent. 
We analyze the error of the steady-state duration of the flow through the following theorem:
\begin{theorem}
\label{theorem:err_t_bound}
The relative error of steady-state duration is bounded. Specifically,
\begin{align}
\hat{\epsilon_{T}} &= |\frac{\hat{T}-\overline{T}}{\overline{T}}| < \theta \quad when \ \Delta R_l(t) < \theta
\end{align}
where $\overline{T}$ and $\hat{T}$ denote the real and estimated flow steady-state duration.
\end{theorem}

We provide the proof in Appendix ~\ref{appendix:proof_error_fct}.

Theorems ~\ref{theorem:err_r_bound} and ~\ref{theorem:err_t_bound} indicate that the steady-state identification algorithm possesses a controllable upper bound error in flow rate and duration of the steady period.

\parab{Guidelines for selecting $\theta$ and $l$.}
The selection of the hyper-parameters $\theta$ and $l$ influences the error and the efficiency of steady-state identification.
Thus, in this section, we analyze the conditions for the values of $\theta$ and $l$.

\parab{\quad Range of $\theta$.}
For accurate identification of the steady-state, the fluctuation threshold within the identification interval should be slightly greater than the rate fluctuation within the steady-state, that is,
\begin{align}
\label{formula:theta_range_def}
\theta &\gtrsim \epsilon_{relative}    
\end{align}
Based on the DCTCP~\cite{dctcp} model, the following formulas are derived:
\begin{align}
D = (W^*+1)\frac{\alpha}{2}, \quad \alpha \approx \sqrt{\frac{2}{W^*}}, \quad W^* = \frac{C \times RTT + K}{N}, 
 \nonumber \\
K > \frac{1}{7}(C \times RTT), \quad T_C = \sqrt{\frac{C\times RTT + K}{2N}} 
 \nonumber
\end{align}
Thus, for $\epsilon_{relative}$, we have
\begin{align}
\epsilon_{relative} &= \frac{D}{(max(W)-min(W))/2} = \frac{(W^*+1)\frac{\alpha}{2}}{(W^*+1)(1-\frac{\alpha}{4})} = \frac{2\alpha}{4-\alpha} \nonumber \\
&\approx \frac{\alpha}{2} \approx \sqrt{\frac{1}{2W^*}} = \sqrt{\frac{N}{2(C \times RTT+K)}} < \sqrt{\frac{7N}{16(C \times RTT)}} \nonumber
\end{align}
When Equation ~\ref{formula:theta_range_def} holds, the boundary for $\theta$ is determined as follows:
\begin{align}
\label{formula:theta_range_res}
\theta &\gtrsim \sqrt{\frac{7N}{16(C \times RTT)}}
\end{align}

When $\theta$ is less than $\epsilon_{relative}$, the steady-state identification algorithm may never detect the steady-state, resulting in no acceleration. 
Conversely, when $\theta$ is excessively large, unsteady-states may be misidentified as steady-states, leading to erroneous rate estimation. 
Hence, $\theta$ should be selected as a value slightly greater than, but close to, the relative rate fluctuation within the steady-state.

\parab{\quad Range of $l$.}
To obtain an accurate $\Delta R_l(t)$ in Equation ~\ref{formula:relative_fluc}, the identification interval should be able to cover at least one period of the steady-state rate presentation, which requires $l$ to be sufficiently large such that:
\begin{align}
\label{formula:l_range_def}
\Delta t(l) &\geq T_C
\end{align}
which we have:
\begin{align}
\label{formula:l_range_res}
\Delta t(l) &= t_l - t_1 \geq T_C > \sqrt{\frac{4C \times RTT}{7N}}
\end{align}
When $l< T_C$, the estimated rate will have a significant deviation. 
When $l$ is excessively large, it will reduce the efficiency of entering the steady-state and decrease the acceleration ratio.
Therefore, $l$ should be set to an appropriate multiple of $T_C$.

\subsection{End of Steady-state}
\label{sec:design:flow}
After flows within a network partition reach a steady-state, we analyze and summarize three types of events that can interrupt the steady-state: 1. entry of new flows; 2. completion of existing flows; 3. rerouting of existing flows (due to link failures or load balancing strategies, \etc). These events dictate at which time point the simulation process for that network partition can jump to, initiating packet-level simulations. These interrupt-type events may be known in advance or may occur in real time,  and we discuss how \sys transitions out of a steady-state in each of these scenarios.

\parab{Predetermined interrupt-type events.}
When testing protocols such as congestion control algorithms or load balancing strategies offline, researchers typically construct traffic demands for a period of time in advance and input them into ns-3~\cite{dcqcn,hpcc,clove}, which initiates the sending of flows according to their start time. 
In addition, the timing of link failures can also be predetermined.  We use a queue to store these events with known timestamps, where the minimum timestamp among them determines the end time of the network partition's steady-state.

\parab{Real-time interrupt-type events.}
When simulating network communications for LLM training or constructing real-time digital twins of networks, the arrival of flows is uncertain, and link failures cannot be known in advance. 
Only when these events are actually injected does \sys become aware of whether the steady-state should be terminated. 
In such scenarios, we present a skip-back mechanism: the simulation process of a network partition initially skips to the nearest known timestamp ($t_1$), when a real-time event occurs with a timestamp ($t_2$), and $t_2 < t_1$, the simulation process skips back to $t_2$, thereby ensuring temporal correctness and preventing simulation errors. 
$\S$\ref{sec:implementation:offset} elucidates further details.
\section{Implementation}\label{sec:implementation}
In this section, we dive into the implementation details of \sys. 
We base \sys on ns-3.17~\cite{hpcc-github} and it consists of thousands of lines of code.
Applying \sys to other PLDES simulators (\eg OMNeT++~\cite{omnet}, DONS~\cite{dons}) is straightforward and ongoing. We also address several practical challenges during the implementation. The implementation details of packet pulsing and timestamp offset are shown in Figure ~\ref{fig:packet-pulsing}.

\begin{figure}[t]
    \centering
    \includegraphics[width=0.98\linewidth]{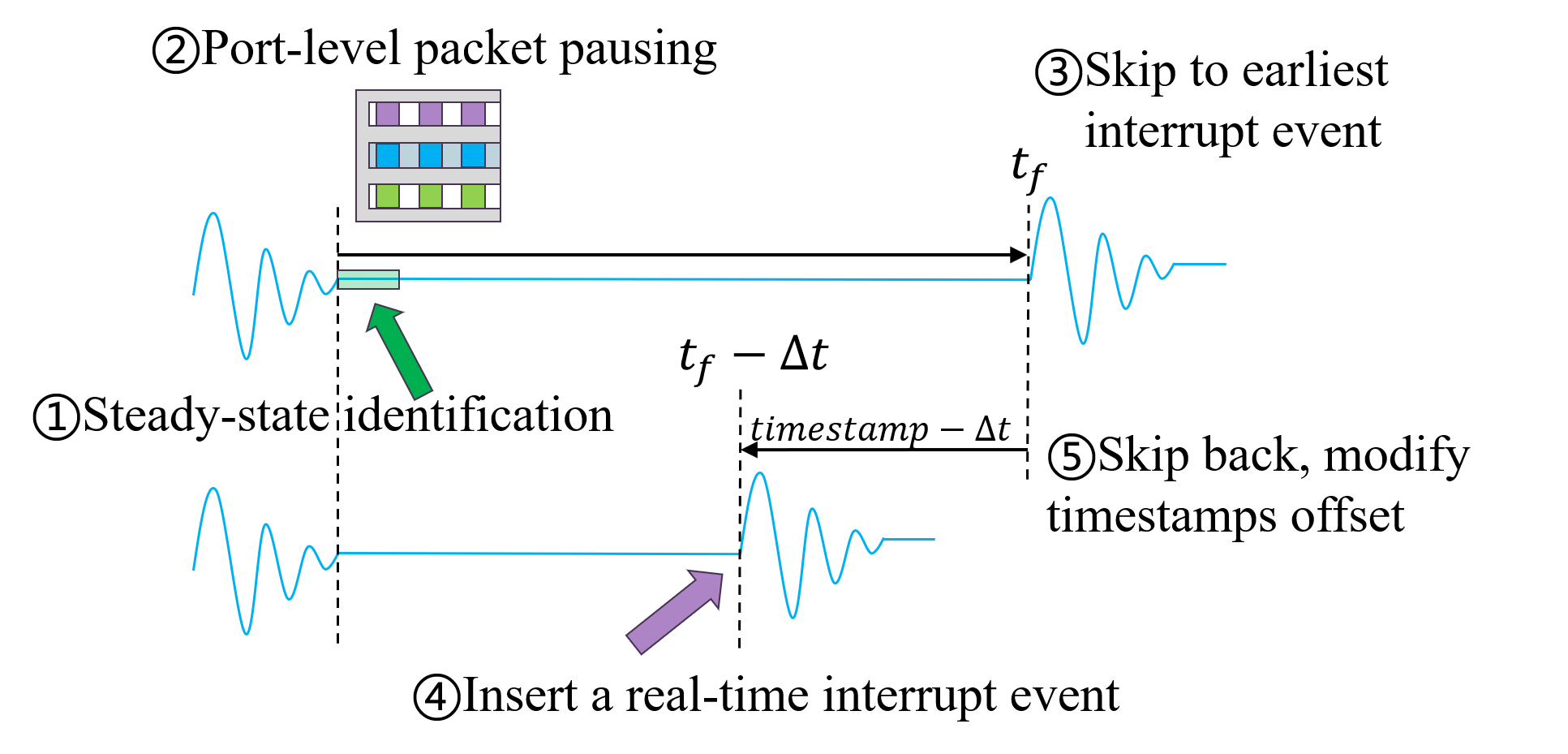}
    \caption{Implementation details of packet pausing and timestamp offsetting.}
    \label{fig:packet-pulsing}
\end{figure}

\subsection{Multi-core Parallelization}
\label{sec:design:parallel}
We find that \sys is fully orthogonal to Unison~\cite{bai2024unison}, the state-of-the-art multithreading optimization for ns-3. 
Unison applies a typical DES parallelization strategy that the system automatically generates multiple logical processes (LPs) given the static network topology. During the simulation process, Unison calculates the load of each LP and adaptively assigns them to multiple threads, running on identical CPU cores. The LP scheduler aims to achieve a relative load balance on multiple cores.


\sys's network partition information can further optimize Unison's parallel efficiency.
To take advantage of this opportunity, we propose a two-stage LP partitioning algorithm. 
Firstly, different LPs are formed according to the network partitions in $\S$\ref{sec:design:algorithm}, since there is no traffic interaction between them, which is more suitable for parallelism with multiple LPs. Unlike Unison, the minimum granularity of the network partitioning is network ports (or interfaces), while the minimum granularity in Unison is switches or hosts.
This fine-grained partitioning reduces the necessity for data synchronization between LPs.

The simulation database implements lightweight concurrency control. Simulation-start queries are parallelized across threads, while simulation-end inserts are coordinated through fine-grained locks. This concurrency control mechanism is enabled by default in the \sysU system in \cref{sec:eval}.

\subsection{Packet Pausing in Switch Port}

The first challenge is to exert steady-states' influence on other parts of the network.
Other nonstable flows certainly belong to different partitions and use different switch ports, so the interaction between these partitions predominantly manifests in the \textit{shared buffer} of the switches. 
For instance, once flows within a network partition stabilize, they might continually occupy a portion of the shared buffer in a switch, diminishing the maximum buffer space available for other ports' flows.
A naive approach is to let these packets be forwarded when the network partition enters a steady-state, because these stable flows will not generate new packets, and then the switch port will not occupy the buffer.
However, inaccuracy will occur within this approach. For example, other flows bypassing the switch may encounter severe congestion and will deplete all shared buffer space. If the buffer size is incorrect, the packet loss timing will be wrong.

We propose an effective scheme to address this challenge. In the steady-state, all metrics within the network partition are almost constant.
Therefore, upon entering steady-state, we pause packet processing for ports in this partition and keep the buffer occupancy constant until exiting steady-state. 
Since these flows continue to occupy the buffer, the maximum buffer size available to other flows will remain the same.

\subsection{Offsetting the Timestamp}
\label{sec:implementation:offset}
The second challenge is to elegantly skip the simulation process within ns-3 without completely reconstructing its underlying architecture.
As \sys has settings for multiple network partitions, each potentially existing in distinct states with varying durations, it is erroneous to effectuate simulation process jumps by altering the global simulation clock. 
Our approach adheres to the maintenance of the regular drive of the simulation clock whereas only adjusting the event timestamps of network partitions while entering a steady-state.
For instance, when a network partition decides to skip to time \(T\) ahead, the timestamps of events (such as packet transmission and message forwarding within switch ports) associated with the specific network partition are increased by \(T\). Meanwhile, the size and sequence number of these flows must also be modified accordingly.

\parab{Skip-back mechanism.}
In the presence of real-time interrupt-type events that cannot be predetermined, \sys will skip the network partition to a known recent interrupt event time (or a relatively large time if there are no predetermined events). 
Subsequently, when the real-time interrupt event is indeed input, \sys then navigates the network partition back to the current time. 
In other words, if a network partition has skipped to time \(T_1\) and necessitates a revert to an earlier time \(T_2\) (\(T_2 < T_1\)).
Prior to the simulation clock reaching \(T_1\), events within this network partition are not processed, rendering the restoration of its state to \(T_2\) uncomplicated. 
\section{Evaluation}\label{sec:eval}
We evaluate \sys in simulation speed, accuracy, and sensitivity. We summarize our results as follows:

\parab{Key Results.}
\begin{icompact}
\item \textbf{Simulation speed:}
\sys is capable of simulating large-scale LLM training workloads.
On a single core, it achieves speedups of 744$\times$ and 510$\times$ for GPT and MoE workloads compared to ns-3. When integrated with parallel DES, it achieves speedups of 1012$\times$ and 716$\times$ for GPT and MoE workloads using 16 CPU cores.
The scalability of \sys is consistent with ns-3.


\item \textbf{Accuracy:} Under various LLM training workloads and CCA scenarios, \sys maintains the average FCT error within 1\% on both GPT and MoE workloads.
\item \textbf{Sensitivity:} 
The monitoring metrics of steady-states identification are nearly equivalent.
\sys exhibits insensitivity across varying hyper-parameters and topologies.

\end{icompact}

\parab{Alternatives.}
We select ns-3~\cite{ns3}, Unison~\cite{bai2024unison}, and flow-level simulator~\cite{simgrid,namyar2024solving} as the comparison alternatives, which represent a range of widely utilized types of simulators.


\begin{table}[t]
\footnotesize
\centering

	\begin{tabular}{c c c c c c}
		\toprule
            \textbf{\# GPUs} & \textbf{\# GPT size, parallel} & \textbf{\# MoE size, parallel} \\
		\midrule
		 64 & 7B, TP8-DP4-PP2 & 8$\times$7B, TP8-EP8-DP4-PP2 \\
         128 & 13B, TP8-DP4-PP4 & 8$\times$13B, TP8-EP8-DP4-PP4 \\
         256 & 22B, TP8-DP8-PP4 & 8$\times$22B, TP8-EP8-DP8-PP4 \\
         1024 & 175B, TP8-DP16-PP8 & 32$\times$22B, TP8-EP8-DP16-PP8 \\
		\bottomrule
	\end{tabular}
	    \caption{Parameters for LLM training workloads.}
	\label{tab:llm_workload}
\end{table}

\parab{Setup.}
The experimental server is a Linux-based server with 2 Intel Xeon CPUs (totaling 56 cores) and 128 GB of memory. 
Simulations conduct 4 Rail Optimized Fat-tree topologies~\cite{nvidia2023superpod} with varying sizes for training GPT-3 models or MoE of different scales for one iteration.
To accurately model the scenario in LLM training networks where multiple NICs on each server may connect to different switches, we represent each GPU as a host in the simulations.

The training workloads consist of groups of DP, PP and EP flows, as existing works on LLM training simulation commonly neglects TP and SP flows~\cite{rashidi2020astrasim, won2023astrasim2, multiverse}
The parallel configurations of GPTs of varying sizes follow the TP(SP)-DP-PP arrangement, whereas MoEs adopt the TP(SP)-EP-DP-PP arrangement, with specific parameters detailed in Table~\ref{tab:llm_workload}.
The micro-batch size is set to be 1, which is the smallest value that still permits pipeline parallelism, yielding a global batch size of DP $\times$ PP.
In all experiments except for the sensitivity analysis, the parameters for \sys are $\theta=5\%$ and $l=2000$.

\subsection{\sys is Fast}
\label{sec:eval:fast}

\begin{figure}[t]
	\centering
	\begin{minipage}{.45\linewidth}
		\subfloat[Speedup under different network sizes]{
		\includegraphics[width=\columnwidth]{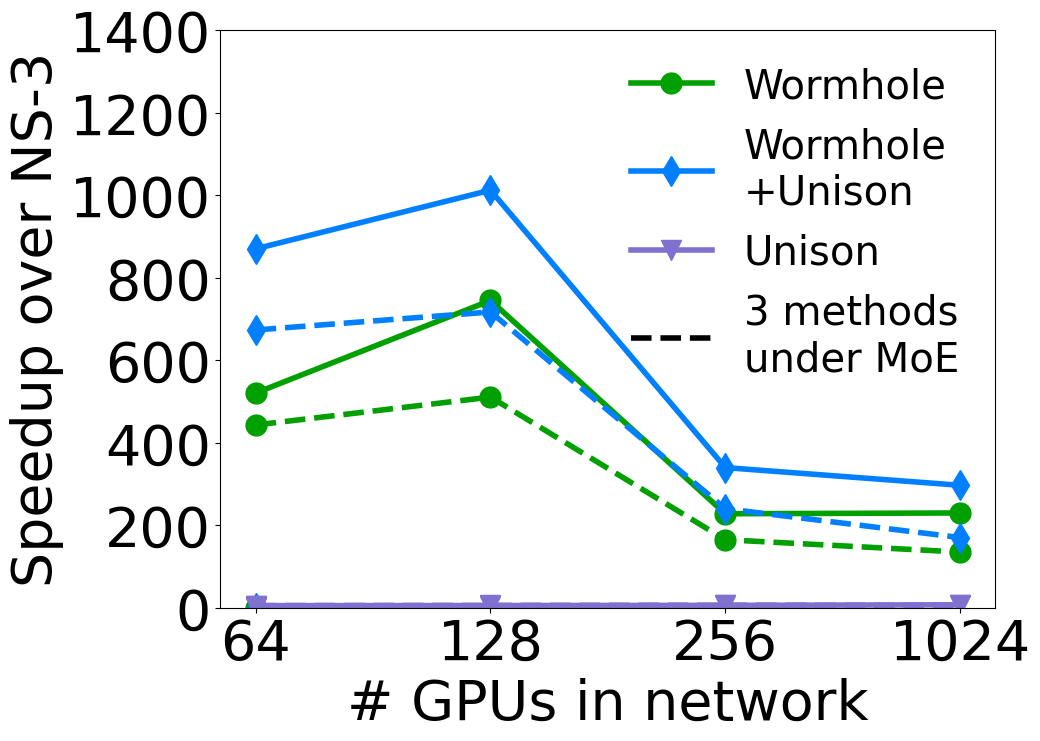} 
		\label{fig:eval:fast:speedup}
		}
	\end{minipage}
	\hspace{0.1in}
	\begin{minipage}{.45\linewidth}
		\subfloat[Speedup under different CCAs]{
		\includegraphics[width=\columnwidth]{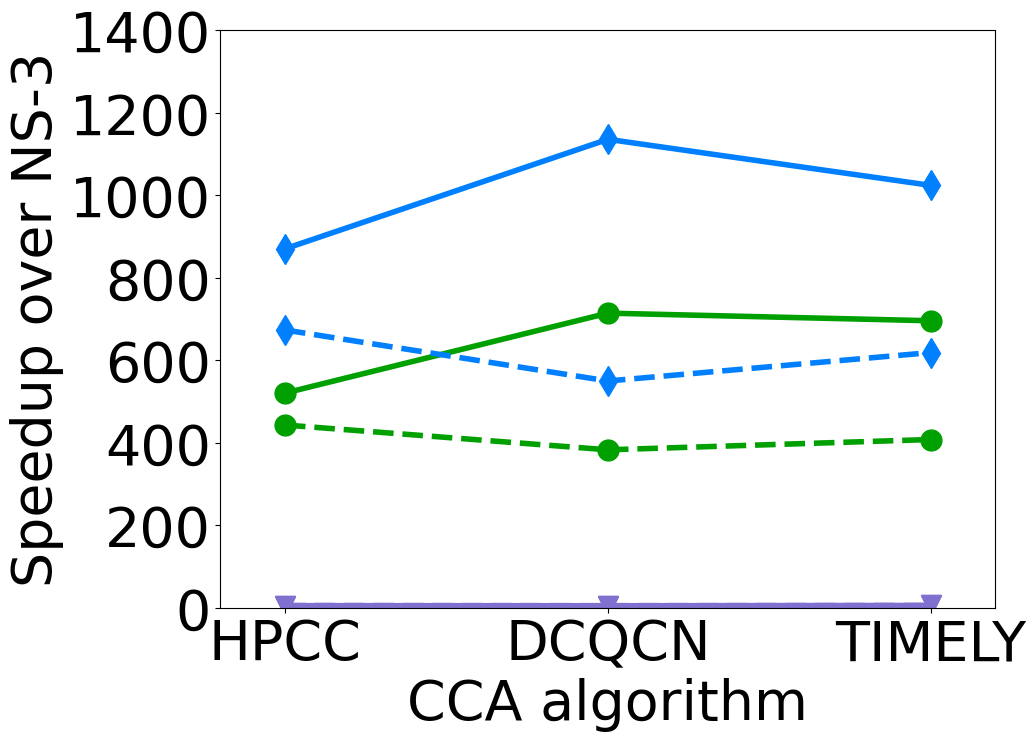} 
		\label{fig:eval:fast:cc}
		}
	\end{minipage}
	\caption{Speedup for simulating LLM training.}
	\label{fig:eval:fast:speed}

        \centering
	\begin{minipage}{.45\linewidth}
		\subfloat[Number of network partitions]{
		\includegraphics[width=\columnwidth]{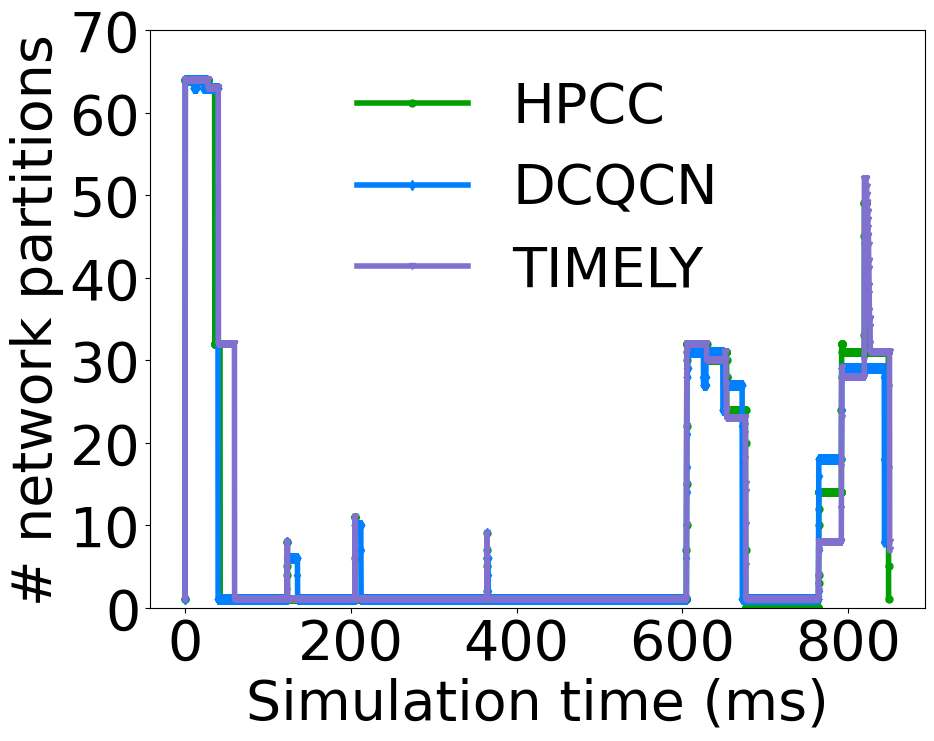} 
		\label{fig:eval:fast:num_segmentation}
		}
	\end{minipage}
	\hspace{0.1in}
	\begin{minipage}{.5\linewidth}
		\subfloat[Storage space of the database]{
		\includegraphics[width=\columnwidth]{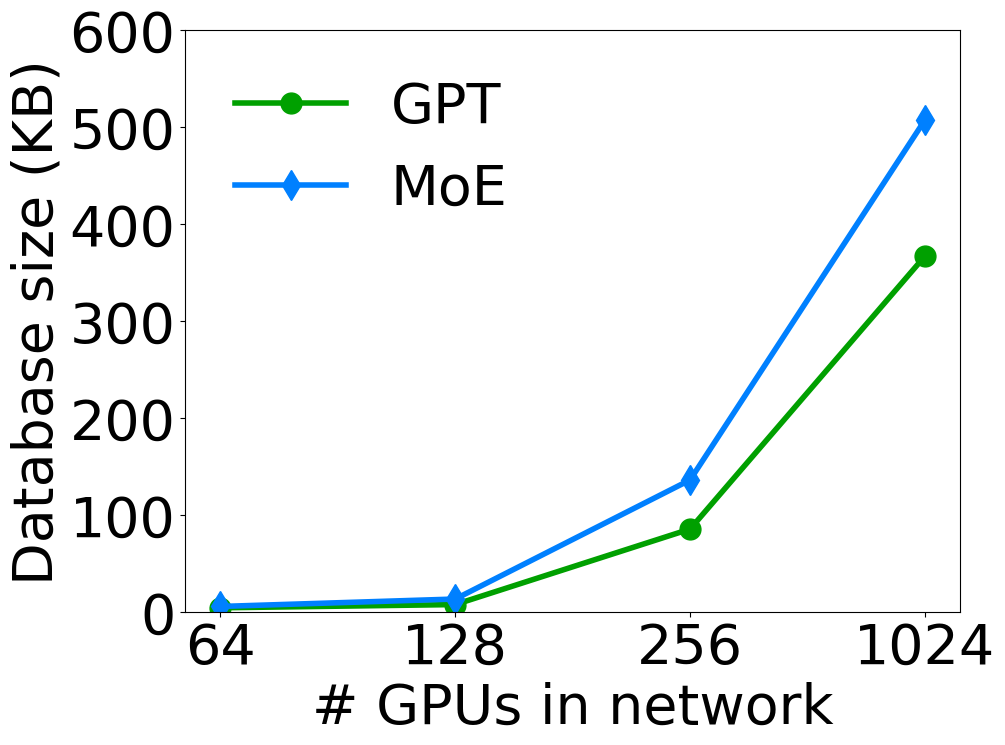} 
		\label{fig:eval:fast:database}
		}
	\end{minipage}
	\caption{Number of network partitions and database size in simulating LLM training.}
	\label{fig:eval:fast:partition_database}
        
	\centering
	\begin{minipage}{.45\linewidth}
		\subfloat[Speedup breakdown]{
		\includegraphics[width=\columnwidth]{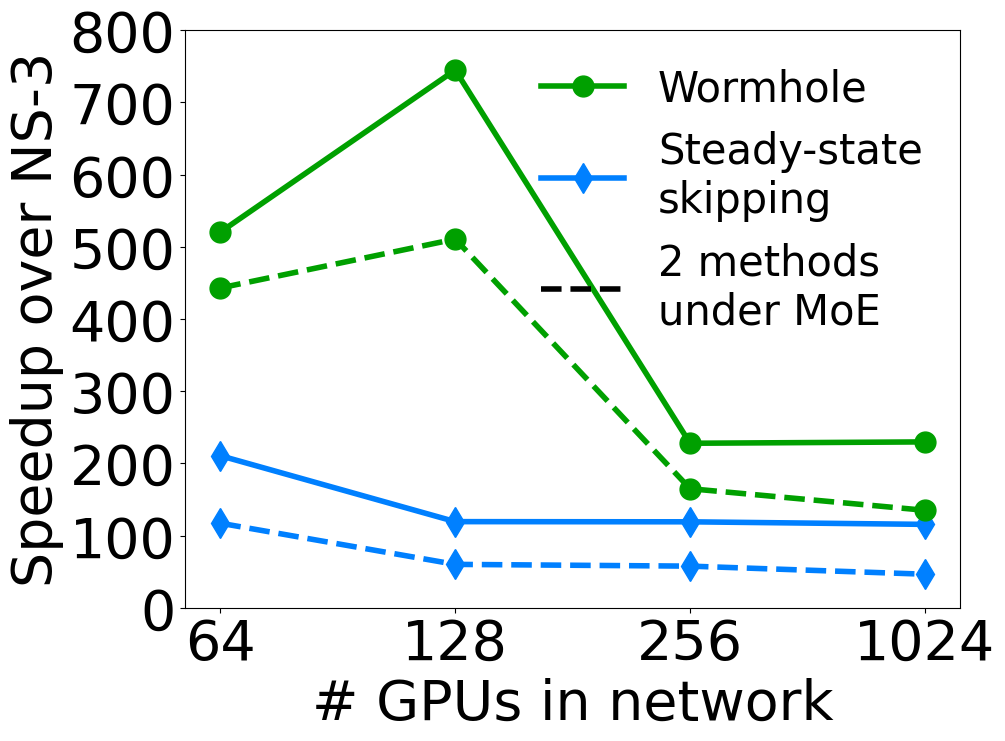} 
		\label{fig:eval:fast:num_steady}
		}
	\end{minipage}
	\begin{minipage}{.48\linewidth}
		\subfloat[Ratio of skipped events]{
		\includegraphics[width=\columnwidth]{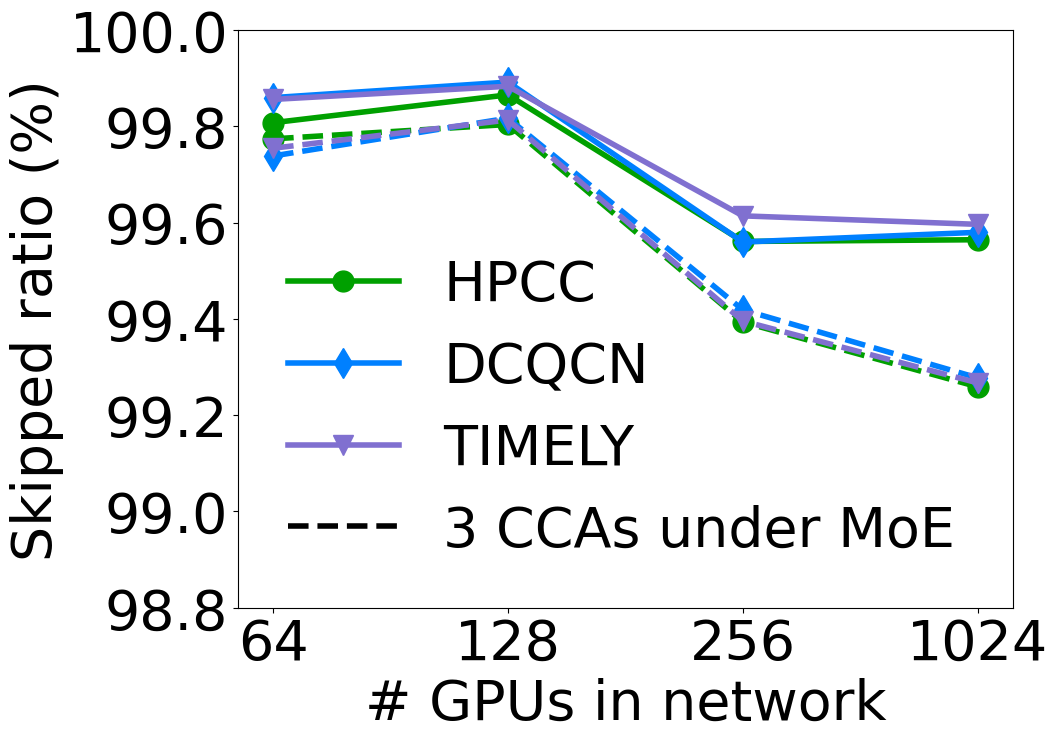} 
		\label{fig:eval:fast:skipped_ratio}
		}
	\end{minipage}
	\caption{Speedup of different mechanism of \sys and ratio of skipped events for simulating LLM training.}
	\label{fig:eval:fast:ablation}  
    
\end{figure}

\begin{figure}[t]
    \centering
	\begin{minipage}{.45\linewidth}
		\subfloat[Average FCT error under different network sizes]{
		\includegraphics[width=\columnwidth]{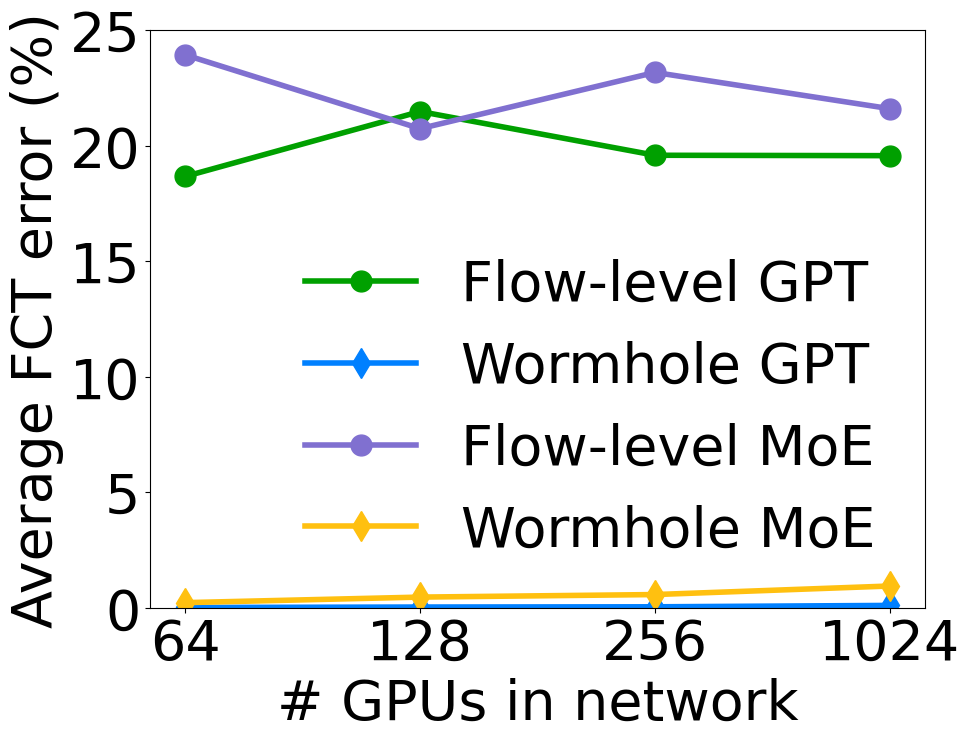} 
		\label{fig:eval:acc:err}
		}
	\end{minipage}
	\hspace{0.1in}
	\begin{minipage}{.45\linewidth}
		\subfloat[Average FCT error under different CCAs]{
		\includegraphics[width=\columnwidth]{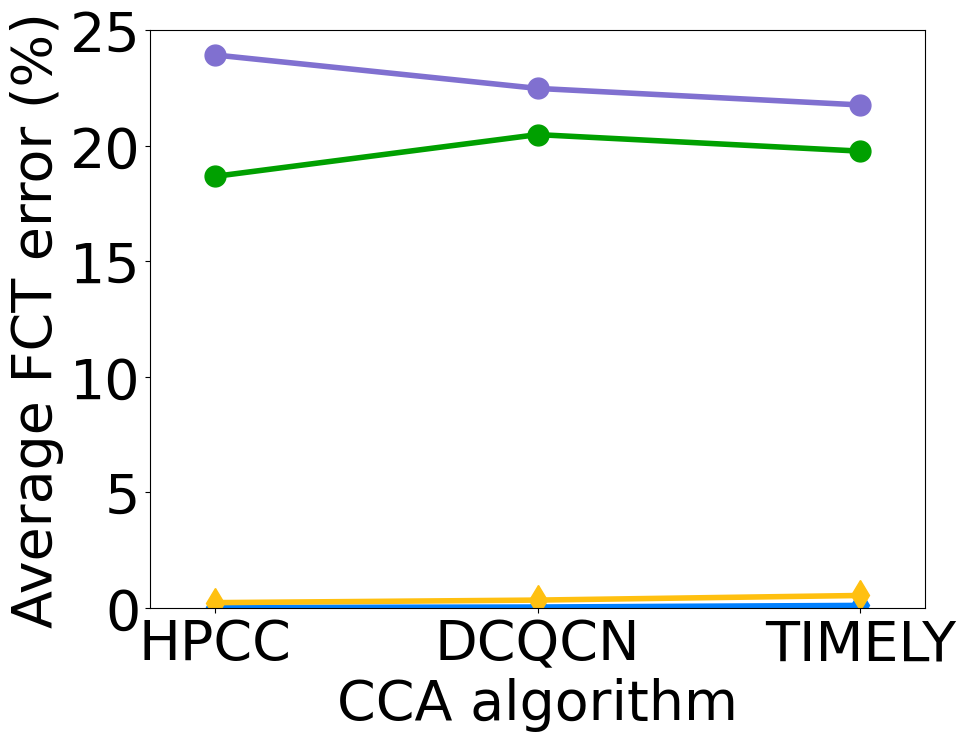} 
		\label{fig:eval:acc:cc}
		}
	\end{minipage}
	\caption{Accuracy for simulating LLM training.}
	\label{fig:eval:acc:error}

\end{figure}

\begin{figure}[t]
	\centering
	\includegraphics[width=0.8\columnwidth]{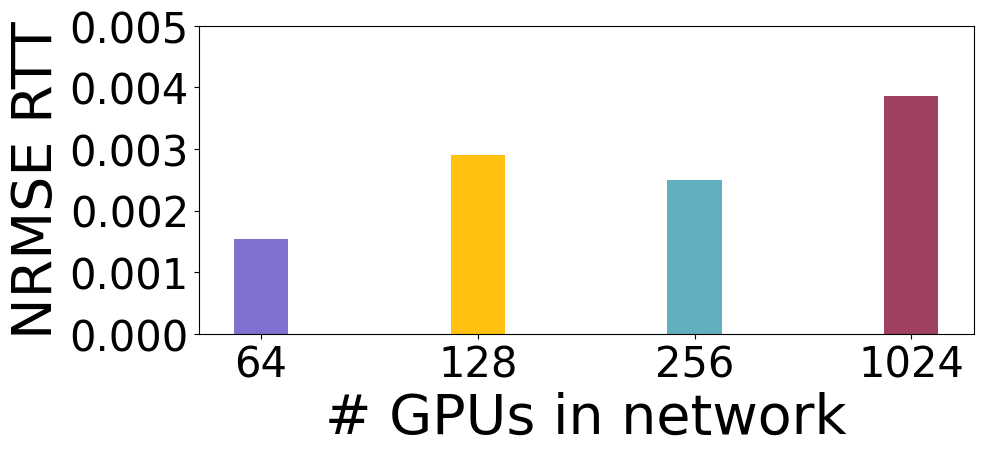}
	\caption{NRMSE of packets RTT in ns-3 and \sys.}
	\label{fig:eval:acc:mse}
\end{figure}

\parab{Speed.}
We evaluate the simulation acceleration of \sys and Unison compared to the original ns-3. 
Figure~\ref{fig:eval:fast:speedup} illustrates the performance of these simulators in LLM training scenarios of varying sizes under HPCC.
While Unison achieves a maximum acceleration of less than 10$\times$, \sys attains an acceleration of 227$\times$-745$\times$ on GPT workloads (135$\times$-510$\times$ on MoE workloads), and \sysU achieves a peak acceleration of 1012$\times$ on GPT workloads (716$\times$ on MoE workloads).
Specifically, \sysU reduces GPT-13B training time on 128 GPUs from 9 hours to 5 minutes, achieving over 1000× speedup.
Figure~\ref{fig:eval:fast:cc} presents the acceleration effects of these simulators under different CCAs in the 64-GPU scenario. 
Results indicate that \sys achieves high acceleration across various CCAs, suggesting that \sys's optimization is broadly applicable for different CCAs.




\parab{Acceleration of memoizing unsteady-states.}
Figure~\ref{fig:eval:fast:skipped_ratio} quantifies the skip rate across CCAs for LLM training, which shows that steady-state fast-forwarding combined with memoization skips >99.5\% of events in GPT workloads and >99.2\% in MoE workloads.  
Comparing the speedup ratios in Figure~\ref{fig:eval:fast:num_steady}, memoization alone delivers an additional 1.93–8.43× acceleration after the dominant steady-state events are skipped, confirming that memoization remains a non-negligible performance booster.

\parab{Acceleration breakdown.}
Figure~\ref{fig:eval:fast:num_steady} shows that skipping steady-states yields speedups exceeding 130× for GPT and 50× for MoE training, confirming that a significant proportion of steady-states in LLM training are redundant and can be safely skipped.
Additionally, the average number of times each flow entered a steady-state was approximately 0.94.
Figure~\ref{fig:eval:fast:skipped_ratio} quantifies the skip rate across CCAs for LLM training, which shows that steady-state fast-forwarding combined with memoization skips >99.5\% of events in GPT workloads and >99.2\% in MoE workloads.  
Comparing the speedup ratios in Figure~\ref{fig:eval:fast:num_steady}, memoization alone delivers an additional 1.93–8.43× acceleration after the dominant steady-state events are skipped, confirming that memoization remains a non-negligible performance booster.
We conduct more acceleration breakdown experiments in Appendix~\ref{appendix:benefit}.

\parab{Simulation database storage cost.}
As illustrated in Figure~\ref{fig:eval:fast:database}, the simulation database stores only a limited amount of critical information, thereby significantly conserving storage space. 
In a scenario with 1024 GPUs, it occupies less than 100KB of space.
Consequently, the database can be entirely put into memory, which enhances the efficiency of database queries and insertions operations.
In addition, the scalability of \sys is consistent with ns-3.

\subsection{\sys is Accurate}
\label{sec:eval:accurate}

\parab{Accuracy.}
We compare the simulation accuracy of \sys and a flow-level simulator relative to the original ns-3 in LLM training workloads. 
Specifically, we calculate the relative error for each flow, and then compute the average of these relative errors. 
Figure~\ref{fig:eval:acc:err} illustrates that \sys maintains an average FCT error within 1\%, significantly lower than the $\sim$20\% error of the flow-level simulator. 
Figure~\ref{fig:eval:acc:cc} presents the average FCT errors of \sys and the flow-level simulator across different CCAs. 
\sys without the unsteady-state memoization mechanism (i.e., skipping only steady-states) exhibits smaller errors than those of \sys itself. 

\parab{Packet-level fidelity.}
We evaluate \sys on boarder metrics. 
\sys evaluation extends to broader metrics through Normalized Root Mean Square Error (NRMSE) computation for all packets of the first flow across scenarios, benchmarked against ns-3. 
Figure~\ref{fig:eval:acc:mse} shows that across multiple scenarios, the NRMSE values fall below 0.005, indicating minimal deviation and confirming end-to-end fidelity across additional performance dimensions.

\begin{figure*}[t]
	\centering
	\begin{minipage}{.32\linewidth}
		\subfloat[{Effects of the monitor metrics}]{
			\includegraphics[width=\columnwidth]{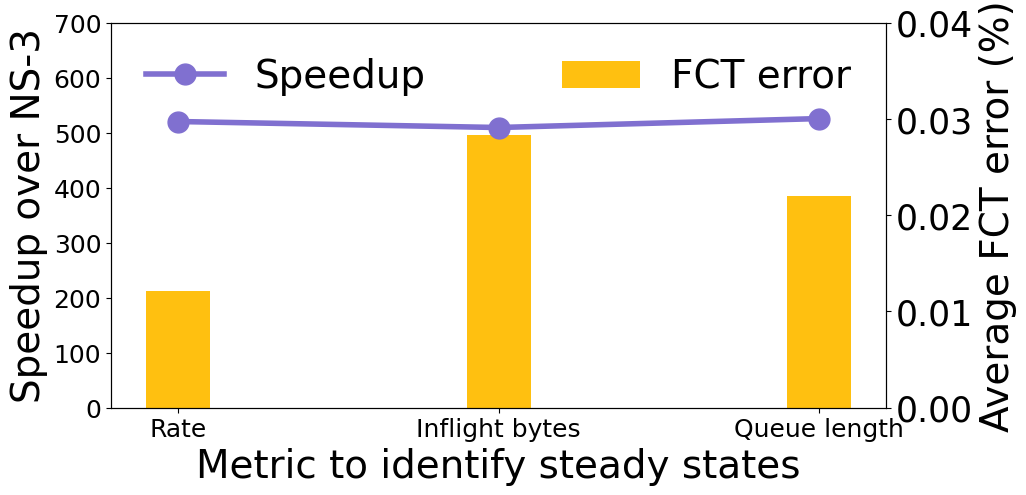}
			\label{fig:eval:sensitive:metrics}
		}
	\end{minipage}
\hfill
	\begin{minipage}{.32\linewidth}
		\subfloat[{Effects of the monitoring interval length $l$}]{
			\includegraphics[width=\columnwidth]{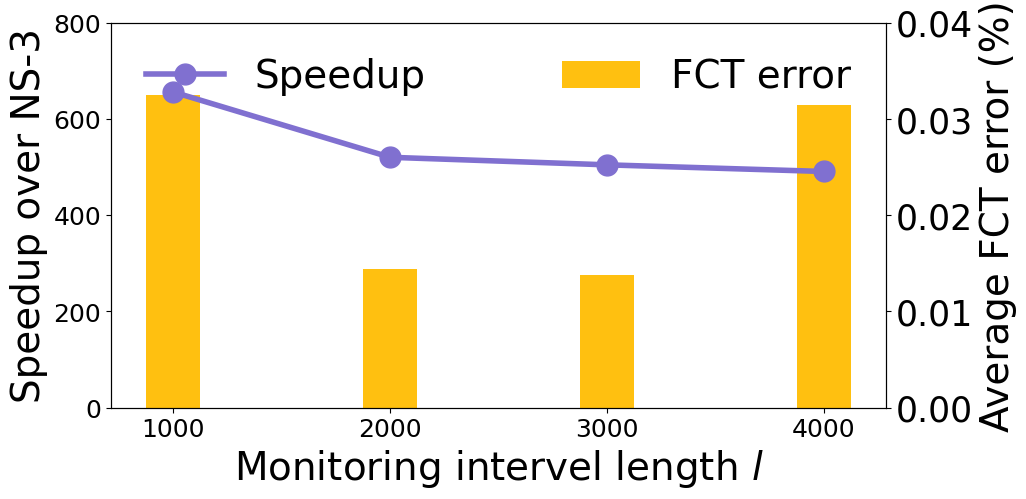}
			\label{fig:eval:sensitive:length}
		}
	\end{minipage}
 \hfill
        \begin{minipage}{.32\linewidth}
		\subfloat[{Effects of the fluctuation threshold $\theta$}]{
			\includegraphics[width=\columnwidth]{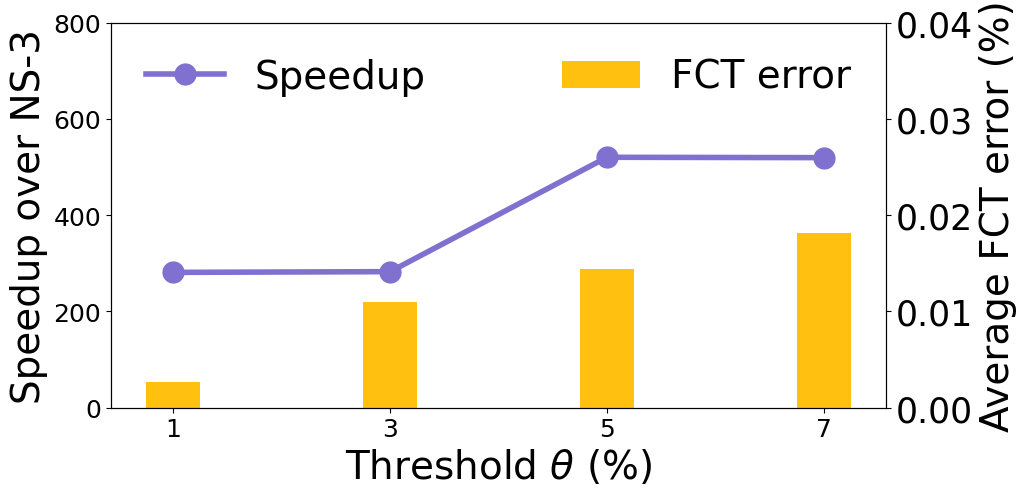}
			\label{fig:eval:sensitive:threshold}
		}
	\end{minipage}
	\caption{Effects of the monitoring metrics and hyper-parameters in \sys.}
	\vspace{-0.1in}
	\label{fig:eval:sensitive}
\end{figure*}

\begin{figure}[t]
	\centering
	\includegraphics[width=0.8\columnwidth]{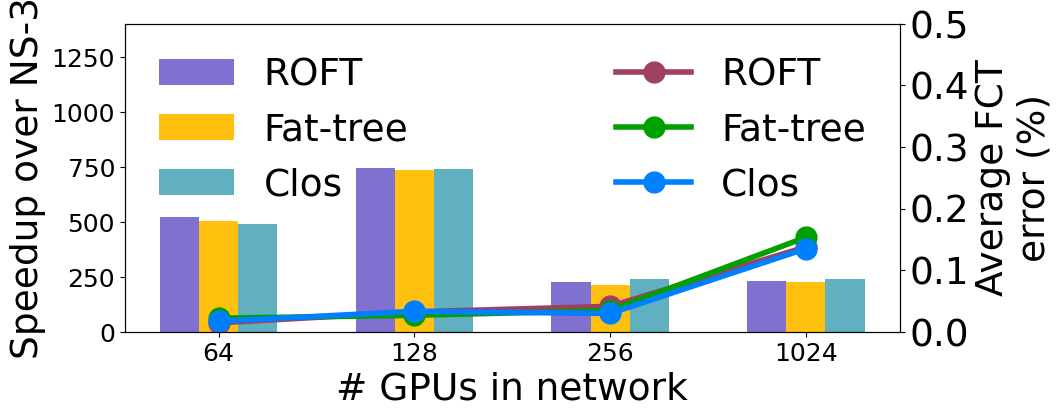}
	\caption{Effects of network topology in \sys.}
	\label{fig:eval:topo}
\end{figure}

\subsection{Sensitivity Analysis}
\label{sec:eval:sensitive}

\parab{Identification metrics.}
The equivalence of sending rate $R$, in-flight bytes $I$, and queue length $Q$ as metrics for steady-state identification algorithms has been theoretically established according to Theorem~\ref{theorem:rate-metrics}. 
To empirically validate this equivalence, we conduct experiments using HPCC with a 128 GPU cluster. 
Specifically, we employ $R$, $I$, and $Q$ individually as metrics for steady-state detection and compare the speedup ratios and errors. 
As depicted in Figure~\ref{fig:eval:sensitive:metrics}, the speedup ratios and errors obtained using these different metrics are closely aligned, which corroborates Theorem~\ref{theorem:rate-metrics}.

\parab{Identification thresholds.} 
By fixing one hyper-parameter and varying the other, we can examine the sensitivity of \sys to two hyper-parameters, $l$ and $\theta$. 
As show in Figures~\ref{fig:eval:sensitive:length}\&\ref{fig:eval:sensitive:threshold}, as $l$ increases or $\theta$ decreases, the criterion for entering the steady-state is more readily satisfied, resulting in a higher ratio of skipped events, which enhances the speedup ratio but also increases the error. 
Therefore, when $l$ and $\theta$ are within appropriate ranges, \sys exhibits insensitivity to the hyper-parameters.
In practice, $\theta=5\%$ is sufficient for most scenarios.

\parab{Network topology.}
We evaluate \sys on multiple network topologies, including the standard setup, Fat-tree~\cite{fattree}, and Clos~\cite{clos1953study}. 
Results are shown in Figure~\ref{fig:eval:topo}, where ROFT denotes Rail-Optimized Fat-tree~\cite{nvidia2023superpod}.
The variation in speedup ratios of \sys across different topologies does not exceed 13\%, and the average FCT error remains within 1\%. 
This indicate that the techniques employed by \sys are applicable to a broad range of data center topologies~\cite{ub-mesh, UBNetworkSimulator, zcube}. 

\subsection{Real-trace based experiments}
\label{sec:eval:real_trace}

We evaluate \sys on a real-world LLM training trace. 
We collect operation-level collective communication latency via NVIDIA Nsight Compute~\cite{nvidia_nsight_compute} from training a GPT-18B on a 256-GPU ROFT cluster. 
The training employs TP8-DP16-PP2-VPP2 parallelism, micro batch size of 1, and global batch size of 512.

\parab{Speed.}
Figure~\ref{fig:eval:real:speed} shows that \sys achieves a 97.75$\times$ speedup over ns-3, while \sysU achieves 133.35$\times$. 
The trace incorporates recomputation and hardware performance fluctuations, producing a more complex workload than idealized traces from SimAI~\cite{simai, sefiane2025mlsynth}.
This complexity reduces \sys speedup relative to idealized scenarios, yet acceleration still approaches more than 100$\times$.

\parab{Accuracy.}
Figure~\ref{fig:eval:real:error} presents that \sys demonstrates 3.02\% end-to-end training time error, while ASTRA-Sim+ns-3~\cite{won2023astrasim2} achieves 3.01\%. 
At comparable accuracy for real-world model training time, \sys provides simulation speed far exceeding the baseline.

\begin{figure}[t]
	\centering
	\begin{minipage}{.45\linewidth}
		\subfloat[Speedup over different simulation methods.]{
		\includegraphics[width=\columnwidth]{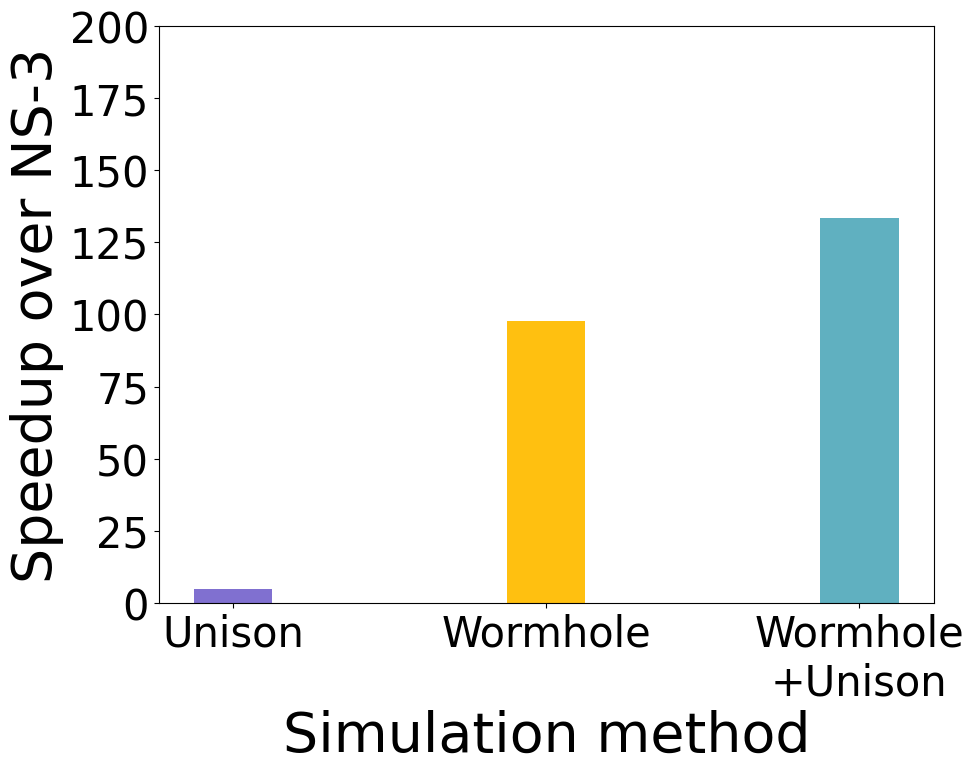} 
		\label{fig:eval:real:speed}
		}
	\end{minipage}
	\hspace{0.1in}
	\begin{minipage}{.45\linewidth}
		\subfloat[End-to-end error over different simulation methods.]{
		\includegraphics[width=\columnwidth]{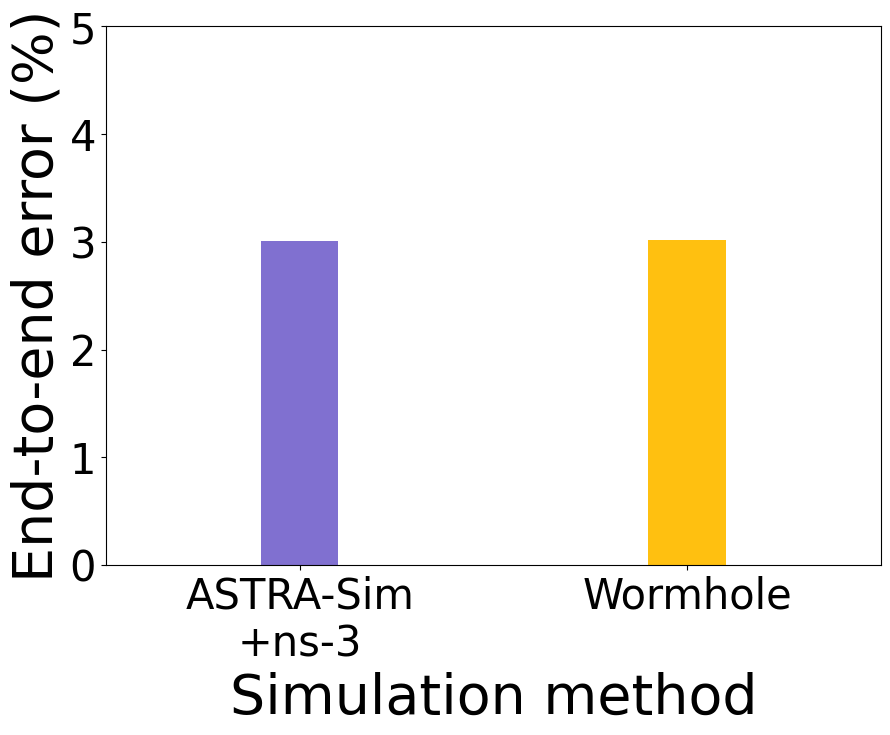} 
		\label{fig:eval:real:error}
		}
	\end{minipage}
	\caption{Speedup and error of real-trace based experiments.}
	\label{fig:eval:real}
\end{figure}



\section{Related Work}
\label{sec:related}
\topparab{Packet-level discrete-event simulators.}
PLDES plays a crucial role in LLM training by providing high-fidelity simulations of network behavior. 
Traditional DES, such as ns-2/3~\cite{ns2,ns3}, OMNeT++~\cite{omnet}, and OPNET~\cite{OPNET}, accurately model packet transmission processes within networks. 
However, as network scale increases, the scalability of these simulators becomes a limiting factor. 
Parallel and distributed DES~\cite{fujimoto1990parallel,jafer2013synchronization,dons,bai2024unison} can enhance the efficiency of DES simulations, but the acceleration effects exhibit sub-linear scaling, eventually reaching an upper bound.

\topparab{LLM training simulation.}
Simulation of large-scale LLM infrastructures has recently gained significant traction. Network simulators such as ASTRA-sim~\cite{rashidi2020astrasim,won2023astrasim2}, SimAI~\cite{simai}, and ATLAHS~\cite{shen2025atlahs} build on PLDES like ns-3 and htsim, but remain bottlenecked by PLDES overhead. To improve scalability, some works~\cite{multiverse,khashab2025nsx,qian2025miniature} exploit GPU or FPGA parallelism, which accelerates execution but leaves the underlying computation unchanged and is orthogonal to \sys. Others adopt coarse-grained models that sacrifice fidelity, either through analytical approximations~\cite{lin2024towards, kundu2024performance} or profiling-based techniques~\cite{feng2024echo, bang2024vtrain, lee2024data, lu2023distsim}.

\topparab{Flow-level simulation.}
Flow-level simulators~\cite{simgrid,mtcloudsim,parsimon,namyar2024solving} employ flow-level abstraction to conduct network simulations and utilize the max-min fairness algorithm~\cite{jaffe1981bottleneck} to address bandwidth allocation. 
This modeling approach achieves a speedup by 2-3 orders of magnitude compared to PLDES, but it introduces significant error margins.
Additionally, a considerable proportion of flow-level simulators~\cite{simgrid,mtcloudsim} focus primarily on specific types of network topologies or applications, thereby limiting their generalizability.

\topparab{AI-based methods.}
AI-based network simulation methods~\cite{yang2022deepqueuenet,kazer18fast,mimicnet,rusek2020routenet,ferriol2023routenet} leverage machine learning to estimate the network performance, resulting in a trade-off with accuracy. They typically require extensive data collection to support model training, but exhibit limited generalizability beyond the training dataset.


\section{Conclusion}
\label{sec:conclusion}
We observe that, in distributed LLM training, packet-level traffic behaviors often exhibit \textit{repetitive contention patterns} and \textit{steady-states}, ignoring these redundant discrete events speeds up the simulation considerably and the error is negligible.
To this end, we propose \sys, a user-transparent PLDES kernel capable of automatically memoization for unsteady-states and skipping for steady-states.
\sys adopts network partitioning, state memoization and reuse, and rate-based steady-state identification to accurately determine the periods of each flow's steady-state.
Experiments demonstrate that \sys can achieve a 744$\times$ speedup over the original ns-3 (510$\times$ for MoE workload), with a bounded error of $<$1\%.
\sysU allows a 1012$\times$ speedup, reducing the simulation time for one GPT-13B training under 128 GPUs from 9 hours to 5 minutes.


\clearpage

{
\bibliographystyle{plain}
\bibliography{sigcomm20}
}

\newpage
\clearpage
\appendix
\section*{Appendix}\label{appendix}
Appendices are supporting material that has not been peer-reviewed.

\section{Network Partitioning Algorithm}
\label{appendix:network_partitioning_algorithm}

Network partitioning algorithm is implemented in Algorithm ~\ref{alg:partition-alg}.

\begin{algorithm}[t]
  \caption{Network partitioning algorithm}  
  \label{alg:partition-alg}
  \SetKwFunction{FMain}{...}
  \SetKwProg{Fn}{Function}{:}{}
   \Fn{construct\_bipartite\_graph{($flows$)}}{
      $connections[1..n + 2 \times m]$ = empty list \\
    \For{$flow$ \textbf{in} $flows$}
    {
        \For{$link$ \textbf{in} $flow.links$}
        {
            $connections[flow.id]$.append($n + link.id$); \\
            $connections[n + link.id]$.append($link.id$); \\
        }
    }
    \textbf{return} $connections$;
    }
    \Fn{DFS{($vertex$, $connections$)}}{
      set\_visited($vertex$); \\
      $local\_partition$ = empty list; \\
    \If{is\_flow($vertex$)}{
        $local\_partition$.append(vertex);
    }
    \For{$neighbor$ \textbf{in} $connections[vertex]$}
    {
        \If{\textbf{not} visited($neighbor$)}
        {
            $neighbor\_partition$ = DFS($neighbor$, $connections$); \\
            $local\_partition$.extend($neighbor\_partition$); \\
        }
    }
    \textbf{return} $local\_partition$;
    }
    \Fn{network\_partitioner{($flows$)}}{
      $connections$ = construct\_bipartite\_graph($flows$); \\
      $network\_partitions$ = empty list; \\
    \For{$flow$ \textbf{in} $flows$}
    {
        \If{\textbf{not} visited($flow.id$)}
        {
            $partitions$ = DFS($flow.id$, $connections$); \\
            $network\_partitions$.append($partitions$); \\
        }
    }      
    \textbf{return} $network\_partitions$;
    }
\end{algorithm}

\section{Proof of Theorem 1}
\label{appendix:proof_metrics}
\begin{proof}
When $R$ is stable, it can be approximated by a fixed value $\overline{R}$ during computation, which is:
\begin{align}
\label{formula:r_approx}
R(t) \approx \overline{R}, \quad |R(t) - \overline{R}| &< \epsilon_R, \quad when \ t_s \leq t \leq t_f, \quad \Delta R < \epsilon_R
\end{align}
This simplification is justified by the negligible variations in the rate during this period, rendering the rate effectively uniform across the interval. 
Such an approximation facilitates subsequent proofs and derivations by reducing the complexity of the analysis and providing a more tractable framework for understanding the system's behavior. 

\parab{CWND.}
For $cwnd$, according to HPCC~\cite{hpcc}, we have 
\begin{align}
cwnd &= R \times \overline{RTT}
\end{align}
where $\overline{RTT}$ represents smooth RTT.
When Equation ~\ref{formula:r_tot_stable} holds, by applying Equation ~\ref{formula:metric_fluc}, we have
\begin{align}
\label{formula:eps_cwnd}
\Delta cwnd &= \max\limits_{t_s \leq t \leq t_f}\{cwnd(t)\} - \min\limits_{t_s \leq t \leq t_f}\{cwnd(t)\} \nonumber \\
&= \overline{RTT} (\max\limits_{t_s \leq t \leq t_f}\{R(t)\} - \min\limits_{t_s \leq t \leq t_f}\{R(t)\}) \nonumber \\ 
&= \overline{RTT} \Delta R \nonumber \\
&< \epsilon_R \overline{RTT} \nonumber \\ 
&:= \epsilon_{cwnd}.
\end{align}
This observation implies that as the rate $R$ enters a steady-state, the fluctuations of $cwnd$ will correspondingly diminish to within another threshold.

\parab{RTT.}
For RTT, TIMELY~\cite{timely} updates the target rate using $\Delta RTT(t) = \frac{dRTT(t)}{dt}$ as follows:
\begin{align}
R(t+\Delta t) = 
\begin{cases}
R(t) + \delta & \text{if} \ \Delta RTT(t) \leq 0 \\
R(t)(1-\beta \frac{dRTT(t)}{dt}) & \text{if}  \ \Delta RTT(t) > 0 
\nonumber
\end{cases}
\end{align}
Using Equation ~\ref{formula:r_approx}, we have
\begin{align}
R(t+\Delta t) - R(t) \nonumber 
\\ =
\begin{cases}
\delta & \text{if} \ \Delta RTT(t) \leq 0 \\
-\beta R(t) \frac{dRTT(t)}{dt}\approx -\beta \overline{R} \Delta RTT(t) & \text{if} \ \Delta RTT(t) > 0 
\nonumber
\end{cases}
\end{align}
when $\Delta RTT(t) = \frac{dRTT(t)}{dt} > 0$, by using Equation ~\ref{formula:r_tot_stable}, we have 
\begin{align}
\Delta RTT(t) &= |\frac{R(t+\Delta t)-R(t)}{\beta \overline{R}}| < \frac{\epsilon_R}{\beta \overline{R}} 
\nonumber
\end{align}
When $\Delta RTT \leq 0$, here $R(t) \leq C$, we have 
\begin{align}
\label{formula:delta_rtt_decrease}
\Delta RTT(t) &= \Delta q(t) = \frac{(R(t) + \delta - C)\Delta t}{C} \leq \frac{\delta \Delta t}{C}
\end{align}
where $C$ is the represents the bandwidth allocated upon convergence, and $q(t)$ is the queuing delay through the bottleneck queue. 
Although Equation ~\ref{formula:delta_rtt_decrease} shows that the fluctuation may be time-related, and it seems that the calculation of $\Delta RTT$ might be related to the cumulative time in the steady period, the actual situation is not the case. 
Considering the speed adjustment mechanism of TIMELY AIMD, during the steady period, an increase in $R$ that causes $\Delta RTT(t) \geq 0$ triggers a Multiplicative Decrease, followed by a series of consecutive Additive Increases until another increase in $R$ again causes $\Delta RTT \geq 0$. 
Therefore, it can be considered that within the steady period, the maximum value of RTT comes from an increase in $R$ that causes $\Delta RTT(t) \geq 0$, and the minimum value of RTT comes from this Multiplicative Decrease adjustment. 
Hence, using Equation ~\ref{formula:delta_rtt_decrease}, there exists a $t'$ such that:
\begin{align}
\Delta RTT &= \Delta RTT(t') < \frac{\delta \Delta t'}{C} \nonumber
\end{align}
where $\Delta t'$ is the interval between two rate adjustments during congestion, which is capped at a specific upper limit. Thus, we have:
\begin{align}
\label{formula:eps_rtt}
\Delta RTT &< max(\frac{\epsilon_R}{\beta \overline{R}}, \frac{\delta \Delta t'}{C}) := \epsilon_{RTT}
\end{align}
for all $t_s \leq t \leq t_f$.
Consequently, $\Delta RTT$ is confined within a constant upper limit. 

\parab{Queue length.}
For queue length $Q$, according to DCTCP~\cite{dctcp}, we have 
\begin{align}
Q(t) &= NW(t) - C \times RTT
\end{align}
where there are $N$ flows on the bottleneck link, and $W(t)$ represents the window size. 
By calculating the fluctuation of $Q$ using Equation ~\ref{formula:metric_fluc}, we have
\begin{align}
\label{formula:eps_q}
\Delta Q &= \max\limits_{t_s \leq t \leq t_f}\{Q(t)\} - \min\limits_{t_s \leq t \leq t_f}\{Q(t)\} 
\nonumber \\
&= \max\limits_{t_s \leq t \leq t_f}\{NW(t) - C \times RTT\} - \min\limits_{t_s \leq t \leq t_f}\{NW(t) - C \times RTT\}
\nonumber \\
&= N\times RTT (\max\limits_{t_s \leq t \leq t_f}\{R(t)\} - \min\limits_{t_s \leq t \leq t_f}\{R(t)\})
\nonumber \\
&= N\times RTT \Delta R
\nonumber \\
&< N\times RTT \epsilon_R
\nonumber \\
&:= \epsilon_{Q}
\end{align}
Consequently, the queue length $Q$ is also stable when $R$ is stable.

\parab{In-flight bytes.}
For in-flight bytes $I$, according to HPCC~\cite{hpcc}, we have 
\begin{align}
I &= Q + R \times RTT_{base}
\end{align}
By calculating the fluctuation of $I$ using Equation (1), we have
\begin{align}
\label{formula:eps_i}
\Delta I &= \max\limits_{t_s \leq t \leq t_f}\{Q(t) + R(t) RTT_{base}\} - \min\limits_{t_s \leq t \leq t_f}\{Q(t) + R(t) RTT_{base}\} 
\nonumber \\
&\leq (\max\limits_{t_s \leq t \leq t_f}\{Q(t)\} - \min\limits_{t_s \leq t \leq t_f}\{Q(t)\}) 
\nonumber \\ 
&\quad + RTT_{base}(\max\limits_{t_s \leq t \leq t_f}\{R(t)\} - \min\limits_{t_s \leq t \leq t_f}\{R(t)\})
\nonumber \\
&= \Delta Q + RTT_{base} \Delta R
\nonumber \\
&< \epsilon_Q + RTT_{base} \epsilon_R
\nonumber \\
&:= \epsilon_I
\end{align}
This implies that when $R$ remains stable, $I$ also maintains stability.

In summary, by Equations ~\ref{formula:eps_cwnd}, ~\ref{formula:eps_rtt}, ~\ref{formula:eps_q} and ~\ref{formula:eps_i}, we have proved Theorem ~\ref{theorem:rate-metrics}.
\end{proof}

\section{Proof of Theorem 2}
\label{appendix:proof_error_rate}
\begin{proof}
During the steady period, by employing Equation ~\ref{formula:relative_fluc}, we obtain
\begin{align}
\label{formula:err_r}
|R(t) - \overline{R}| &\leq |\max\limits_{1 \leq k \leq l}\{R(t_k)\} - \min\limits_{1 \leq k \leq l}\{R(t_k)\}| = \hat{R} \Delta R_l(t) < \theta \hat{R}
\end{align}
for every single $R(t)$ in the steady period. 
Consequently, using Equation ~\ref{formula:err_r}, the error in estimating $\overline{R}$ using $\hat{R}$ is given by:
\begin{align}
\hat{\epsilon_R} &= |\frac{\hat{R}-\overline{R}}{\overline{R}}| = |\frac{\sum\limits_{k=1}^{l} (R(t_k)-\overline{R})}{l\overline{R}}| \leq \frac{\sum\limits_{k=1}^{l} |R(t_k)-\overline{R}|}{l\overline{R}} < |\frac{l\theta \hat{R}}{l\overline{R}}| = \frac{\theta \hat{R}}{\overline{R}}
\nonumber
\end{align}
Thus we have
\begin{align}
|\frac{\hat{R}}{\overline{R}} - 1| < \theta \frac{\hat{R}}{\overline{R}}
\nonumber
\end{align}
By solving $\frac{\hat{R}}{\overline{R}}$, we have 
\begin{align}
\label{formula:r_rel_range}
\frac{1}{1 + \theta} < \frac{\hat{R}}{\overline{R}} < \frac{1}{1 - \theta}
\end{align}
Thus, applying Equation ~\ref{formula:r_rel_range} we have 
\begin{align}
\hat{\epsilon_{R}} &= |\frac{\hat{R}}{\overline{R}} - 1| < \text{max}(\frac{\theta}{1 + \theta}, \ \frac{\theta}{1 - \theta}) = \frac{\theta}{1 - \theta}
\nonumber
\end{align}
which demonstrates Theorem 2.
\end{proof}

\section{Proof of Theorem 3}
\label{appendix:proof_error_fct}
\begin{proof}
Assuming the remaining data volume of the flow upon entering the steady-state is $F$, we have the following for the actual and estimated cases:
\begin{align}
F &= \frac{\int_{t_1}^{t_l} R(t)dt}{t_l - t_1} \overline{T} = \overline{R} \overline{T} \nonumber \\
F &= \frac{\sum\limits_{k=1}^{l}R(t_k)}{l} \hat{T} = \hat{R} \hat{T} \nonumber
\end{align}
By combining these two equations, we obtain
\begin{align}
\hat{T} &= \frac{\overline{R}}{\hat{R}}\overline{T} \nonumber
\end{align}
Using Equation ~\ref{formula:r_rel_range}, we have 
\begin{align}
1 - \theta < \frac{\overline{R}}{\hat{R}} < 1 + \theta
\nonumber
\end{align}
Consequently, we have
\begin{align}
\hat{\epsilon_T} &= |\frac{\hat{T} - \overline{T}}{\overline{T}}| = |\frac{\hat{T}}{\overline{T}} - 1| = |\frac{\overline{R}}{\hat{R}} - 1| < \theta
\nonumber
\end{align}
which demonstrates Theorem 3.
\end{proof}

\newpage
\section{Incremental Recalculate the Network Partitions}
\label{appendix:incremental}

\begin{algorithm}[t]
    \caption{Incremental Partitioning Algorithm}
    \label{alg:incremental-par}
    \SetKwFunction{FMain}{...}
    \SetKwProg{Fn}{Function}{:}{}
    \Fn{on\_new\_flow\_enter{($newflow$)}}{
        $affected\_partitions$ = empty list; \\
        \For{$partition$ \textbf{\text{in}} $all\_partitions$}{
            \If{$newflow.path \text{ } \cap \text{ }partition \neq \emptyset$}{
                $affected\_partition$.append($partition$);
            }
        }
        \If{$affected\_partitions$.size $= 0$}{
            $all\_partitions$.append(partition($newflow$));
        }
        \ElseIf{$affected\_partitions$.size $= 1$}{
            $affected\_partition$.append($newflow$);
        }
        \Else{
            $network\_partitioner$($affected\_partitions.flow$);
        }
    }
    \Fn{on\_old\_flow\_leave($oldflow$)}{
        $affected\_flows$ = empty list; \\
        \For{$flow$ \textbf{\text{in}} $partition.flows$}{
            \If{$flow \textbf{ is not } oldflow$}{
                $affected\_flows$.append($flow$);
            }
        }
        \If{$affected\_flows$.size $\leq$ 1}{
            $partition$.erase($oldflow$); \\
            \If{$affected\_flows$.size = 0}{
                $all\_partitions$.erase($partition$);
            }
        }
        \Else{
            $network\_partitioner$($affected\_flows$);   
        }
    }
\end{algorithm}

The third challenge is to reconstruct the network partitions in a reasonable time and possibly speed up the recalculation process. Since \sys follows port-level partitioning, we observe that such a recalculation will be triggered if a new flow enters the network or a flow is about to finish and leave. \sys keeps track of all flows both activated and scheduled to enter so that during the simulation, \sys knows when the recalculation should be done.

It is worth noting that in a typical scenario of LLM training, such recalculation of partitions may occur frequently, whereas a certain rearrangement requires extra computing resources. Therefore, we propose a more efficient way to update partitions. Notice that many of the updates occur in the local part of the network topology and only affect a few other flows (or partitions). To take advantage of it, we follow a strategy of incremental updating. In other words, merely those affected parts will be updated. More details are in Appendix \ref{appendix:incremental}.

The incremental algorithm is highly restricted in local parts rather than the overall topology. Hence, the worst case of the algorithm is degrading into applying Algorithm~\ref{alg:partition-alg} on the basis of the entire network.

The incremental updating algorithm is implemented in Algorithm~\ref{alg:incremental-par}.
The Algorithm is described in the following two cases.

\parab{Addition of new flow.} If a new flow comes into the topology, we first compute the set of partitions \textit{affected} on account of the specific new flow's path. By affected, meaning the new flow passes through the domain of some existing partitions. After that, different operations are carried out based on the number of partitions that are affected. More details are in Algorithm~\ref{alg:incremental-par}.

\parab{Completion of flow.} Whenever an about-to-finish flow leaves the network, its own partition may split into several parts if the only connection is the leaving flow. In this case, we calculate the remaining number of flows in the leaving flow's partition and do operations based on the exact number of remaining flows. If the remaining number is larger than 1, we use the Function \textit{network\_partitioner} in Algorithm~\ref{alg:partition-alg} to reconstruct the local part of network. Otherwise, simply delete the flow in its partition.

\section{Speedup over simulation duration of \sys}
\label{appendix:benefit}

\begin{figure}[t]
	\centering
	\includegraphics[width=0.8\columnwidth]{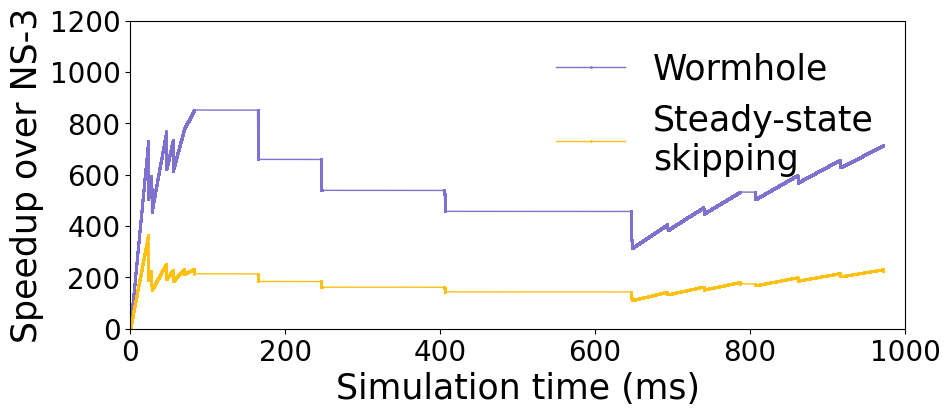}
	\caption{Benefit of \sys over simulation progress.}
	\label{fig:eval:fast:benefit}
\end{figure}

Figure~\ref{fig:eval:fast:benefit} illustrates the temporal progression of Wormhole speedup ratio over simulation time for a 64-GPU configuration. 
The speedup metric quantifies the quotient of events processed by ns-3 divided by events processed by \sys. 
DP flow scenarios with larger flow sizes and complex workload in the initial and final phases amplify \sys performance advantage, while PP flow scenarios with smaller flow sizes in the intermediate phase reduce the average speedup. 
Over time, memoization accumulates performance benefits.

\end{document}